\documentclass[11pt,oneside,article,a4paper]{memoir}
\usepackage{amsmath}
\usepackage{amsfonts}
\usepackage{amssymb}
\usepackage{bm}
\usepackage[english]{babel}
\usepackage[latin1]{inputenc}
\usepackage[T1]{fontenc}
\usepackage{graphicx}
\usepackage{mathtools}
\usepackage{mathrsfs}
\usepackage{cancel}
\usepackage[amsmath,thmmarks]{ntheorem}
\usepackage{color}
\usepackage{rotating}
\newsubfloat{figure}

\title{{\scshape\textbf{Existence of Travelling Wave Solutions to the Maxwell-Pauli and Maxwell-Schrödinger Systems}}}
\author{Kim Petersen and Jan Philip Solovej}
\date{\small Department of Mathematical Sciences, University of Copenhagen,\\
Universitetsparken 5, DK-2100 Copenhagen, Denmark,\\
Email adresses: kp@math.ku.dk and solovej@math.ku.dk\\
\vspace{0.2cm}\copyright\ 2013 by the authors}

\abstractrunin
\abslabeldelim{.}

\theoremstyle{plain}
\theoremseparator{.}
\newtheorem{saetning}{Theorem}

\newtheorem{lemma}[saetning]{Lemma}
\newtheorem{proposition}[saetning]{Proposition}

\theoremstyle{plain}
\theorembodyfont{\normalfont}
\theoremseparator{.}
\newtheorem{definition}[saetning]{Definition}
\newtheorem{bemaerkning}[saetning]{Remark}

\theoremstyle{nonumberplain}
\theorembodyfont{\normalfont}
\theoremsymbol{\ensuremath{\square}}
\theoremseparator{.}
\newtheorem{proof}{Proof}

\setsecheadstyle{\normalfont\scshape}

\begin{document}

\maketitle

\begin{abstract}
We study two mathematical descriptions of a charged particle interacting with it's self-generated electromagnetic field. The first model is the one-body Maxwell-Schrödinger system where the interaction of the spin with the magnetic field is neglected and the second model is the related one-body Maxwell-Pauli system where the spin-field interaction is included. We prove that there exist travelling wave solutions to both of these systems provided that the speed $|\bm{v}|$ of the wave is not too large. Moreover, we observe that the energies of these solutions behave like $\frac{m\bm{v}^2}{2}$ for small velocities of the particle, which may be interpreted as saying that the effective mass of the particle is the same as it's bare mass.
\end{abstract}
\chapter{Introduction}
Consider a single spin-$\frac{1}{2}$ particle of mass $m>0$ and charge $Q\in\mathbb{R}\setminus\{0\}$ interacting with it's self-generated electromagnetic field -- the Maxwell-Schrödinger system in Coulomb gauge says\footnote{We use Gaussian units.} that
\begin{align}
\begin{split}\label{MS}
\Box \bm{A}_{\mathrm{t}}&=\frac{4\pi}{c} P\bm{J}_{\mathrm{S}}[\psi_{\mathrm{t}},\bm{A}_{\mathrm{t}}],\\
i\hbar\partial_t\psi_{\mathrm{t}}&=\Bigl(\frac{1}{2m}\nabla_{\mathrm{S},\bm{A}_{\mathrm{t}}}^2+\mathscr{E}_{\mathrm{EM}}[\bm{A}_{\mathrm{t}},\partial_t\bm{A}_{\mathrm{t}}]\Bigr)\psi_{\mathrm{t}},\\
\mathrm{div}\bm{A}_{\mathrm{t}}&=0.
\end{split}
\end{align}
where $\hbar>0$ is the reduced Planck constant, $\psi_{\mathrm{t}}(t):\mathbb{R}^3\to \mathbb{C}^2$ is the quantum mechanical wave function describing the particle, $\bm{A}_{\mathrm{t}}(t):\mathbb{R}^3\to\mathbb{R}^3$ is the classical magnetic vector potential induced by the particle, with $c>0$ denoting the speed of light we let $\nabla_{\mathrm{S},\bm{A}_{\mathrm{t}}}=i\hbar\nabla+\frac{Q}{c}\bm{A}_{\mathrm{t}}$ denote the covariant derivative with respect to $\bm{A}_{\mathrm{t}}$, $\Box=\frac{1}{c^2}\partial_t^2-\Delta$ is the d'Alembertian, $P=1-\nabla\mathrm{div}\Delta^{-1}$ is the Helmholtz projection onto the solenoidal subspace of divergence free vector fields, $\mathscr{E}_{\mathrm{EM}}[\bm{A}_{\mathrm{t}},\partial_t\bm{A}_{\mathrm{t}}](t)$ is the energy
\begin{align*}
\mathscr{E}_{\mathrm{EM}}[\bm{A}_{\mathrm{t}},\partial_t\bm{A}_{\mathrm{t}}](t)=\frac{1}{8\pi}\int_{\mathbb{R}^3}\Bigl(|\nabla\times \bm{A}_{\mathrm{t}}(t)(\bm{y})|^2+\Bigl|\frac{1}{c}\partial_t\bm{A}_{\mathrm{t}}(t)(\bm{y})\Bigr|^2\Bigr)\,\mathrm{d}\bm{y}
\end{align*}
associated with the electromagnetic field and $\bm{J}_{\mathrm{S}}[\psi_{\mathrm{t}},\bm{A}_{\mathrm{t}}](t):\mathbb{R}^3\to\mathbb{R}^3$ denotes the probability current density given by
\begin{align*}
\bm{J}_{\mathrm{S}}[\psi_{\mathrm{t}},\bm{A}_{\mathrm{t}}](t)(\bm{x})=-\frac{Q}{m}\mathrm{Re}\bigl\langle\psi_{\mathrm{t}}(t)(\bm{x}),\nabla_{\mathrm{S},\bm{A}_{\mathrm{t}}}\psi_{\mathrm{t}}(t)(\bm{x})\bigr\rangle.
\end{align*}
Here, $\langle\cdot,\cdot\rangle$ is the usual inner product in $\mathbb{C}^2$ and $|\cdot|$ denotes the norm induced by this inner product. We will also study an alternative (more accurate) description of the physical system's time evolution that takes the interactions between the magnetic field and the quantum mechanical spin of the particle into account. By letting $\bm{\sigma}$ denote the $3$-vector with the Pauli matrices
\begin{align*}
\sigma^1=\begin{pmatrix}0&1\\1&0\end{pmatrix},\quad \sigma^2=\begin{pmatrix}0&-i\\i&0\end{pmatrix}\quad\textrm{and}\quad \sigma^3=\begin{pmatrix}1&0\\0&-1\end{pmatrix}
\end{align*}
as components we can write the Maxwell-Pauli system in Coulomb gauge as
\begin{align}
\begin{split}\label{MP}
\Box \bm{A}_{\mathrm{t}}&=\frac{4\pi}{c} P\bm{J}_{\mathrm{P}}[\psi_{\mathrm{t}},\bm{A}_{\mathrm{t}}],\\
i\hbar\partial_t\psi_{\mathrm{t}}&=\Bigl(\frac{1}{2m}\nabla_{\mathrm{P},\bm{A}_{\mathrm{t}}}^2+\mathscr{E}_{\mathrm{EM}}[\bm{A}_{\mathrm{t}},\partial_t\bm{A}_{\mathrm{t}}]\Bigr)\psi_{\mathrm{t}},\\
\mathrm{div}\bm{A}_{\mathrm{t}}&=0.
\end{split}
\end{align}
Here, $\nabla_{\mathrm{P},\bm{A}_{\mathrm{t}}}$ is short for $\bm{\sigma}\cdot\nabla_{\mathrm{S},\bm{A}_{\mathrm{t}}}$ whose square by definition is the Pauli operator -- the Lichnerowicz formula says that this operator can alternatively be expressed as
\begin{align}
\nabla_{\mathrm{P},\bm{A}_{\mathrm{t}}}^2=\nabla_{\mathrm{S},\bm{A}_{\mathrm{t}}}^2-\frac{\hbar Q}{c}\bm{\sigma}\cdot\nabla\times\bm{A}_{\mathrm{t}}.\label{Pauli}
\end{align}
The probability current density $\bm{J}_{\mathrm{P}}[\psi_{\mathrm{t}},\bm{A}_{\mathrm{t}}](t):\mathbb{R}^3\to\mathbb{R}^3$ is given by
\begin{align*}
\bm{J}_{\mathrm{P}}[\psi_{\mathrm{t}},\bm{A}_{\mathrm{t}}](t)(\bm{x})=-\frac{Q}{m}\mathrm{Re}\bigl\langle \psi_{\mathrm{t}}(t)(\bm{x}),\bm{\sigma}\nabla_{\mathrm{P},\bm{A}_{\mathrm{t}}}\psi_{\mathrm{t}}(t)(\bm{x})\bigr\rangle.
\end{align*}

In the literature, the Maxwell-Schrödinger system often refers to the following equations in $\psi_{\mathrm{t}}(t):\mathbb{R}^3\to \mathbb{C}^2$, $\bm{A}_{\mathrm{t}}(t):\mathbb{R}^3\to \mathbb{R}^3$ and $\varphi_{\mathrm{t}}(t):\mathbb{R}^3\to \mathbb{R}$
\begin{align}
\begin{split}\label{literature}
-\Delta\varphi_{\mathrm{t}}-\frac{1}{c}\partial_t\mathrm{div}\bm{A}_{\mathrm{t}}&=4\pi Q|\psi_{\mathrm{t}}|^2,\\
\Box\bm{A}_{\mathrm{t}}+\nabla \Bigl(\frac{1}{c}\partial_t\varphi_{\mathrm{t}}+\mathrm{div}\bm{A}_{\mathrm{t}}\Bigr)&=\frac{4\pi}{c}\bm{J}_{\mathrm{S}}[\psi_{\mathrm{t}},\bm{A}_{\mathrm{t}}],\\
i\hbar\partial_t\psi_{\mathrm{t}}&=\Bigl(\frac{1}{2m}\nabla_{\mathrm{S},\bm{A}_{\mathrm{t}}}^2+Q\varphi_{\mathrm{t}}\Bigr)\psi_{\mathrm{t}}.
\end{split}
\end{align}
This system approximates the quantum field equations for an electrodynamical nonrelativistic many-body system. When expressed in Coulomb gauge it reads
\begin{align}
\begin{split}\label{literatureCoulomb}
\Box\bm{A}_{\mathrm{t}}&=\frac{4\pi}{c}P\bm{J}_{\mathrm{S}}[\psi_{\mathrm{t}},\bm{A}_{\mathrm{t}}],\\
i\hbar\partial_t\psi_{\mathrm{t}}&=\Bigl(\frac{1}{2m}\nabla_{\mathrm{S},\bm{A}_{\mathrm{t}}}^2+Q^2\bigl(|\bm{x}|^{-1}*|\psi_{\mathrm{t}}|^2\bigr)\Bigr)\psi_{\mathrm{t}},\\
\mathrm{div}\bm{A}_{\mathrm{t}}&=0,
\end{split}
\end{align}
which only deviates from \eqref{MS} by the absence of the term $\mathscr{E}_{\mathrm{EM}}[\bm{A}_{\mathrm{t}},\partial_t\bm{A}_{\mathrm{t}}]\psi_{\mathrm{t}}$ (making no difference for the existence question studied in this paper) and by the presence of $Q^2\bigl(|\bm{x}|^{-1}*|\psi_{\mathrm{t}}|^2\bigr)\psi_{\mathrm{t}}$ on the right hand side of the Schrödinger equation. The nonlinear term $Q^2\bigl(|\bm{x}|^{-1}*|\psi_{\mathrm{t}}|^2\bigr)\psi_{\mathrm{t}}$ should be perceived as a mean field originating from the Coulomb interactions between the particles present in the many-body system -- when the system only consists of a single particle there simply are no other particles to interact with and so \eqref{MS} is a better decription than \eqref{literatureCoulomb} of the one-body system. In \cite{Coclite}, Coclite and Georgiev observe that there do not exist any nontrivial solutions in the form $(\psi_{\mathrm{t}},\bm{A}_{\mathrm{t}},\varphi_{\mathrm{t}})(t)(\bm{x})=(\mathrm{e}^{-i\omega t}\psi(\bm{x}),\bm{0},\varphi(\bm{x}))$ to the system \eqref{literature} expressed in Lorenz gauge -- they also prove that such solutions do exist when one adds an attractive potential of Coulomb type to the Schrödinger equation. The analogous problem in a bounded space region has been studied by Benci and Fortunato \cite{Benci}. Several authors have studied the existence of solitary solutions to other systems than \eqref{MS} and \eqref{MP}. For example Esteban, Georgiev and Séré \cite{Esteban} prove the existence of stationary solutions to the Maxwell-Dirac system in Lorenz gauge -- in the same paper they also treat the Klein-Gordon-Dirac system. The existence of travelling wave solutions to a certain nonlinear equation describing the dynamics of pseudo-relativistic boson stars in the mean field limit has been proven by Fröhlich, Jonsson and Lenzmann \cite{Enno} and also the existence of solitary water waves has been studied extensively -- let us mention the recent paper by Buffoni, Groves, Sun and Wahlén \cite{Wahlen}. Finally, we mention that the well-posedness of the initial value problem associated with \eqref{literature} expressed in different gauges has been subject to a lot of research -- see \cite{Ta,GNS,NT,NWG} and references therein. In \cite{KP1}, the unique existence of a local solution to the many-body Maxwell-Schrödinger initial value problem expressed in Coulomb gauge is proven. For $j\in\{\mathrm{S},\mathrm{P}\}$ the aim of the present paper is to show that
\begin{align}
\begin{split}\label{MSP}
\Box \bm{A}_{\mathrm{t}}&=\frac{4\pi}{c} P\bm{J}_{j}[\psi_{\mathrm{t}},\bm{A}_{\mathrm{t}}],\\
i\hbar\partial_t\psi_{\mathrm{t}}&=\Bigl(\frac{1}{2m}\nabla_{j,\bm{A}_{\mathrm{t}}}^2+\mathscr{E}_{\mathrm{EM}}[\bm{A}_{\mathrm{t}},\partial_t\bm{A}_{\mathrm{t}}]\Bigr)\psi_{\mathrm{t}},\\
\mathrm{div}\bm{A}_{\mathrm{t}}&=0
\end{split}
\end{align}
admits solutions $(\psi_{\mathrm{t}},\bm{A}_{\mathrm{t}})$ in the form
\begin{align}
\begin{split}\label{travelform}
\psi_{\mathrm{t}}(t)(\bm{x})&=\mathrm{e}^{-i\omega t}\psi(\bm{x}-\bm{v}t),\\
\bm{A}_{\mathrm{t}}(t)(\bm{x})&=\bm{A}(\bm{x}-\bm{v}t),
\end{split}
\end{align}
with $\omega\in\mathbb{R}$, $\bm{v}\in\mathbb{R}^3$ and both of the functions $\psi$ and $\bm{A}$ defined on $\mathbb{R}^3$. As time evolves the shapes of these functions do not change -- the initial states $\psi$ and $\bm{A}$ are simply translated in space with constant velocity $\bm{v}$ (and in case $\omega\neq 0$ the phase of the wave function oscillates too). For this reason solutions in the form \eqref{travelform} are often called \emph{travelling waves}. 
\begin{figure}[h!]
	\centering
		\includegraphics[width=0.17\textwidth]{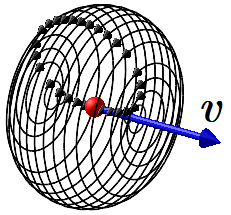}
	\caption{A travelling wave solution models the situation where a particle travels in space at a constant velocity $\bm{v}$ and it's self-generated electromagnetic field travels along with it.}
	\label{fig:TravellingWave}
\end{figure}
To formulate our main theorem ensuring the existence of travelling wave solutions to \eqref{MSP} we let $H^1$ denote the usual Sobolev space of order $1$ and introduce the space $D^1$ of locally integrable functions $A$ on $\mathbb{R}^3$ that have distributional first order derivatives in $L^2$ and vanish at infinity, in the sense that the (Lebesgue-)measure of the set
\begin{align*}
\{\bm{x}\in\mathbb{R}^3\,\mid\, |A(\bm{x})|>t\}
\end{align*}
is finite for all $t>0$. The elements in the space $D^1$ satisfy the Sobolev inequality $\|A\|_{L^6}\leq K_S\|\nabla A\|_{L^2}$ and by equipping $D^1$ with the inner product $(A,B)\mapsto(\nabla A,\nabla B)_{L^2}$ we obtain a Hilbert space in which $C_0^\infty$ is a dense subspace. Also, for $\lambda>0$ we define the quantities
\begin{align*}
\Theta_{j,\pm}^{\lambda}=\begin{cases}\pm c&\textrm{if }j=\mathrm{S},\\-\frac{8\pi K_S^3 Q^2\lambda}{\hbar}\pm\sqrt{\frac{(8\pi)^2K_S^6Q^4\lambda^2}{\hbar^2}+c^2}&\textrm{if }j=\mathrm{P}\end{cases}.
\end{align*}
Our main theorem then asserts the following.
\begin{saetning}\label{main}
For all $j\in\{\mathrm{S},\mathrm{P}\}$, $\lambda>0$ and $\bm{v}\in\mathbb{R}^3$ with $0<|\bm{v}|<\Theta_{j,+}^{\lambda}$ there exist $\omega\in\mathbb{R}$ and functions $(\psi,\bm{A})\in H^1\times D^1$ satisfying $\|\psi\|_{L^2}^2=\lambda$ such that $(\psi_{\mathrm{t}},\bm{A}_{\mathrm{t}})(t)(\bm{x})=\bigl(\mathrm{e}^{-i\omega t}\psi(\bm{x}-\bm{v}t),\bm{A}(\bm{x}-\bm{v}t)\bigr)$ solves \eqref{MSP}.
\end{saetning}
\begin{bemaerkning}
In quantum mechanics the quantity $\|\psi\|_{L^2}^2$ is interpreted as the total probability of the particle being located somewhere in space. Therefore $\lambda=1$ is the physically interesting case.
\end{bemaerkning}
We do not prove any uniqueness results concerning the travelling wave solutions, but in Theorem \ref{energyact} we show that the energies of the solutions produced by the proof of Theorem \ref{main} behave like $\frac{m\bm{v}^2}{2}\lambda$ for small $|\bm{v}|$, meaning that the effective mass of the particle equals it's bare mass. Here, the energy of a (sufficiently nice) solution $(\psi_{\mathrm{t}},\bm{A}_{\mathrm{t}})$ to \eqref{MSP} refers to the inner product $\bigl(\psi_{\mathrm{t}},\mathscr{H}_j\bigl(\bm{A}_{\mathrm{t}},\frac{\partial_t\bm{A}_{\mathrm{t}}}{4\pi c^2}\bigr)\psi_{\mathrm{t}}\bigr)_{L^2}$, 
where
\begin{align}
\mathscr{H}_j\bigl(\bm{A},-\tfrac{P\bm{E}}{4\pi}\bigr)=\frac{1}{2m}\nabla_{j,\bm{A}}^2+\mathscr{E}_{\mathrm{EM}}[\bm{A},-c^2P\bm{E}].\label{HamiltonianQM}
\end{align}
is the quantum mechanical (electromagnetic potential-dependent) Hamiltonian of the system. In \cite{KP1,KP3}, we have motivated the expression for \eqref{HamiltonianQM} in the case $j=\mathrm{S}$. For any given normalized state $\psi$ the Hamilton equations associated with the classical Hamiltonian $\bigl(\bm{A},-\tfrac{P\bm{E}}{4\pi}\bigr)\mapsto \bigl(\psi,\mathscr{H}_j\bigl(\bm{A},-\tfrac{P\bm{E}}{4\pi}\bigr)\psi\bigr)_{L^2}$ defined on the symplectic manifold $PH^1\times PL^2$ say that
\begin{align}
\frac{1}{c^2}\partial_t\bm{A}_{\mathrm{t}}(t)=-P\bm{E}_{\mathrm{t}}(t)\textrm{ and }-\partial_tP\bm{E}_{\mathrm{t}}(t)=\Delta\bm{A}_{\mathrm{t}}(t)+\frac{4\pi}{c}P\bm{J}_j[\psi,\bm{A}_{\mathrm{t}}(t)].\label{Hamiltoneq}
\end{align}
In light of \eqref{Hamiltoneq}'s first equation it is natural to represent the energy of a given solution $(\psi_{\mathrm{t}},\bm{A}_{\mathrm{t}})$ by the average of $\mathscr{H}_j$ evaluated at the point $\bigl(\bm{A}_{\mathrm{t}},\frac{\partial_t\bm{A}_{\mathrm{t}}}{4\pi c^2}\bigr)$. Observe also that the operator acting on the right hand side of \eqref{MSP}'s second equation is exactly $\mathscr{H}_j\bigl(\bm{A}_{\mathrm{t}},\frac{\partial_t\bm{A}_{\mathrm{t}}}{4\pi c^2}\bigr)$ and that replacing $\psi$ in \eqref{Hamiltoneq} by the time-dependent wave function $\psi_{\mathrm{t}}$ produces the first equation in \eqref{MSP}. Note that the energy of any solution $(\psi_{\mathrm{t}},\bm{A}_{\mathrm{t}})$ to \eqref{MSP} with $\|\psi_{\mathrm{t}}\|_{L^2}=1$ is a conserved quantity -- in particular, the energy of a travelling wave solution as in theorem \ref{main} is given by
\begin{align}
E_j(\bm{v},\psi,\bm{A})=\frac{1}{2m}\|\nabla_{j,\bm{A}}\psi\|_{L^2}^2+\frac{1}{8\pi}\int_{\mathbb{R}^3}\Bigl(|\nabla\times\bm{A}|^2+\Bigl|\Bigl(\frac{\bm{v}}{c}\cdot\nabla\Bigr)\bm{A}\Bigr|^2\Bigr)\,\mathrm{d}\bm{x}\lambda.
\end{align}

The paper is organized as follows: In Section \ref{variation} we show that Theorem \ref{main} can be proven by minimizing a certain functional. This functional is shown to be bounded from below under suitable conditions in Section \ref{bddnessbelow}, whereby it is meaningful to consider the functional's infimum under those conditions. In Section \ref{propertiesofIsection} we investigate the properties of the infimum and in Section \ref{existenceofminim} the infimum is shown to be attained by proving a variant of the concentration-compactness principle of Lions \cite{Lions1,Lions2}. Finally, in Section \ref{behaviorofenergy} we consider the behavior of the physical system's energy for small velocities of the particle.

\section*{Acknowledgements}
JPS thanks Jakob Juul Stubgaard for discussions about the effective mass.

\chapter{Formulation as a Variational Problem}\label{variation}
As a natural step towards proving Theorem \ref{main} we plug the travelling wave expressions \eqref{travelform} into \eqref{MSP}, resulting in the system of equations
\begin{align}
\begin{split}\label{internMS2}
\Bigl(\frac{1}{c^2}(\bm{v}\cdot\nabla)^2-\Delta\Bigr)\bm{A}&=\frac{4\pi}{c} P\bm{J}_j[\psi,\bm{A}],\\
-\hbar(\theta+i\bm{v}\cdot\nabla)\psi&=\frac{1}{2m}\nabla_{j,\bm{A}}^2\psi,\\
\mathrm{div}\bm{A}&=0
\end{split}
\end{align}
on $\mathbb{R}^3$, where we have set $\theta=\frac{1}{\hbar}\mathscr{E}_{\mathrm{EM}}[\bm{A},(\bm{v}\cdot\nabla)\bm{A}]-\omega$. The existence of a solution to \eqref{internMS2} can be proven by finding a minimum point -- or any other type of stationary point for that matter -- of the functional
\begin{align}
\mathscr{E}^{\bm{v}}_{j}(\psi,\bm{A})=\frac{1}{2m}\bigl\|\nabla_{j,\bm{A}}\psi\bigr\|_{L^2}^2+\frac{1}{8\pi}\Bigl(\|\nabla\otimes \bm{A}\|_{L^2}^2-\Bigl\|\Bigl(\frac{\bm{v}}{c}\cdot\nabla\Bigr)\bm{A}\Bigr\|_{L^2}^2\Bigr)\nonumber\\
+\bigl(\psi,i\hbar\bm{v}\cdot\nabla\psi\bigr)_{L^2},\label{Eexp1}
\end{align}
on the set
\begin{align*}
\mathcal{S}_\lambda=\bigl\{(\psi,\bm{A})\in H^1\times D^1 \, \big| \, \|\psi\|_{L^2}^2=\lambda,\mathrm{div}\bm{A}=0\bigr\},
\end{align*}
where $\nabla\otimes \bm{A}$ denotes a $9$-vector with the first derivatives $\partial_{x^j}A^k$ as components ($j,k\in\{1,2,3\}$). To prove this we will use the boundedness of $P$ as an operator on $L^p$ for all $p\in (1,\infty)$, which follows from the Mikhlin multiplier theorem \cite{Mikhlin} since any function $\bm{p}\mapsto \frac{\bm{p}^\beta}{\bm{p}^2}$ with $|\beta|=2$ is contained in $C^\infty(\mathbb{R}^3\setminus\{\bm{0}\})$ with
\begin{align*}
\Bigl|\partial^\alpha\Bigl(\frac{\bm{p}^\beta}{\bm{p}^2}\Bigr)\Bigr|\leq \frac{C_{\alpha,\beta}}{|\bm{p}|^{|\alpha|}}
\end{align*}
for any multi index $\alpha$ and all $\bm{p}\in\mathbb{R}^3\setminus\{\bm{0}\}$.
\begin{lemma}\label{Mintosolve}
Let $\bm{v}\in\mathbb{R}^3$, $\lambda>0$ and $j\in\{\mathrm{S},\mathrm{P}\}$ be given. Then any minimizer $(\psi,\bm{A})$ of $\mathscr{E}_j^{\bm{v}}$ on $\mathcal{S}_{\lambda}$ solves \eqref{internMS2} for some $\theta\in \mathbb{R}$.
\end{lemma}
\begin{proof}
Suppose that $\mathscr{E}_j^{\bm{v}}$ has a minimum point $(\psi,\bm{A})$ on $\mathcal{S}_{\lambda}$. Consider also some function $\Psi\in C_0^\infty$ as well as an arbitrary real valued $C_0^\infty$-vector field $\bm{a}$. Then $P\bm{a}$ is divergence free and contained in $D^1$ (in fact, in all positive exponent Sobolev spaces), so the functions $f_{\Psi}$ and $g_{\bm{a}}$ given on an open interval containing $0$ by
\begin{align*}
f_{\Psi}:\varepsilon\mapsto\mathscr{E}_j^{\bm{v}}\left(\frac{(\psi+\varepsilon\overline{\Psi})\sqrt{\lambda}}{\|\psi+\varepsilon\overline{\Psi}\|_{L^2}},\bm{A}\right)\quad\textrm{respectively}\quad g_{\bm{a}}:\varepsilon\mapsto\mathscr{E}_j^{\bm{v}}(\psi,\bm{A}+\varepsilon P\bm{a})
\end{align*}
have local minima at $\varepsilon=0$. Now, set $\theta=-\frac{\|\nabla_{j,\bm{A}}\psi\|_{L^2}^2+2m(\psi,i\hbar\bm{v}\cdot\nabla\psi)_{L^2}}{2m\hbar\lambda}$ and observe that the mappings $f_{\Psi}$ and $g_{\bm{a}}$ are both differentiable at $0$ with derivatives
\begin{align}
\frac{\mathrm{d}f_{\Psi}}{\mathrm{d}\varepsilon}(0)=2\mathrm{Re}\Bigl\langle \frac{1}{2m}\nabla_{j,\bm{A}}^2\psi+\hbar\theta\psi+i\hbar\bm{v}\cdot\nabla\psi,\Psi\Bigr\rangle_{\mathscr{D}'}\label{fjdq}
\end{align}
and
\begin{align}
&\frac{\mathrm{d}g_{\bm{a}}}{\mathrm{d}\varepsilon}(0)\nonumber\\
&=\!\int_{\mathbb{R}^3}\!\Bigl(\!-\frac{1}{c}P\bm{a}\cdot \bm{J}_j[\psi,\bm{A}]+\!\frac{1}{4\pi}\sum_{k=1}^3\partial_kP\bm{a}\cdot \partial_k\bm{A}\!-\!\frac{1}{4\pi c^2}(\bm{v}\cdot\nabla)P\bm{a}\cdot(\bm{v}\cdot\nabla)\bm{A}\Bigr)\mathrm{d}\bm{x}\nonumber\\
&=\Bigl\langle -\frac{1}{c}P\bm{J}_j[\psi,\bm{A}]-\frac{1}{4\pi}\Delta\bm{A}+\frac{1}{4\pi c^2}(\bm{v}\cdot\nabla)^2\bm{A},\bm{a}\Bigr\rangle_{\mathscr{D}'}.\label{gjdq}
\end{align}
To obtain the expression for $\frac{\mathrm{d}g_j^{\bm{a}}}{\mathrm{d}\varepsilon}(0)$ we have here used the fact that
\begin{align*}
\int_{\mathbb{R}^3}(1-P)\bm{b}\cdot P\bm{K}\,\mathrm{d}\bm{x}=\int_{\mathbb{R}^3}P\bm{b}\cdot(1-P)\bm{K}\,\mathrm{d}\bm{x}=0
\end{align*}
for any choice of fields $\bm{b}\in C_0^\infty$ and $\bm{K}\in L^p$ with $p\in (1,2]$. Since the functions $f_{\Psi}$, $f_{i\Psi}$ and $g_{\bm{a}}$ have local minima at $\varepsilon=0$ we are in position to conclude that $(\psi,\bm{A})$ solves \eqref{internMS2}.
\end{proof}
We will finish this section by making two important observations concerning the functional $\mathscr{E}_j^{\bm{v}}$. First of all, it is sometimes useful to rewrite the expression \eqref{Eexp1} by using the Hermiticity of the Pauli matrices and the general matrix identity
\begin{align}
(\bm{\sigma}\cdot \bm{F})(\bm{\sigma}\cdot \bm{G})=(\bm{F}\cdot\bm{G})\mathrm{id}_{2\times 2}+i\bm{\sigma}\cdot(\bm{F}\times\bm{G})\label{sigmavector}
\end{align}
to obtain
\begin{align}
\mathscr{E}^{\bm{v}}_{j}(\psi,\bm{A})=\frac{1}{2m}\bigl\|\nabla_{j,\bm{A}+\frac{mc}{Q}\bm{v}}\psi\bigr\|_{L^2}^2-\frac{Q}{c}\bigl(\psi,\bm{v}\cdot\bm{A}\psi\bigr)_{L^2}-\frac{m\bm{v}^2}{2}\lambda\nonumber\\
+\frac{1}{8\pi}\Bigl(\|\nabla\otimes \bm{A}\|_{L^2}^2-\Bigl\|\Bigl(\frac{\bm{v}}{c}\cdot\nabla\Bigr)\bm{A}\Bigr\|_{L^2}^2\Bigr)\label{Eexp2}
\end{align}
for all $(\psi,\bm{A})\in \mathcal{S}_{\lambda}$. Secondly, any element $O$ in the rotation group $\mathcal{SO}(3)$ gives rise to the identity
\begin{align*}
\mathscr{E}^{\bm{v}}_{j}(\psi,\bm{A})=\mathscr{E}^{O\bm{v}}_{j}\bigl(U_O\circ\psi\circ O^{-1},O\circ \bm{A}\circ O^{-1}\bigr)\textrm{ for }(\psi,\bm{A})\in\mathcal{S}_{\lambda},
\end{align*}
where $U_O$ is one of the two elements in the preimage of $\{O\}$ under the double cover $\mathcal{SU}(2)\to \mathcal{SO}(3)$ defined by mapping $U\in \mathcal{SU}(2)$ to the matrix representation with respect to the basis $(\sigma^1,\sigma^2,\sigma^3)$ of the endomorphism $M\mapsto UMU^*$ on the space of Hermitean, traceless matrices. Hence, we can without loss of generality think of $\bm{v}$ as pointing, say, in the $x^1$-direction.

\chapter{Boundedness from below}\label{bddnessbelow}
At this point we have defined our main goal, namely to minimize the functional $\mathscr{E}_j^{\bm{v}}$ on the set $\mathcal{S}_{\lambda}$. In order for this task to even make sense $\mathscr{E}_j^{\bm{v}}$ of course has to be bounded from below on $\mathcal{S}_{\lambda}$. In special cases -- e.g. for $\bm{v}=\bm{0}$ -- the question about boundedness from below is trivially answered affirmatively, but it turns out that $\mathscr{E}_j^{\bm{v}}$ is \emph{not} in general bounded from below on $\mathcal{S}_{\lambda}$. 
\begin{proposition}\label{lowerb}
For all $j\in\{\mathrm{S},\mathrm{P}\}$, $\lambda>0$ and $\bm{v}\in\mathbb{R}^3$ with sufficiently large length the functional $\mathscr{E}_j^{\bm{v}}$ is unbounded from below on $\mathcal{S}_{\lambda}$. On the other hand for any $j\in\{\mathrm{S},\mathrm{P}\}$, $\lambda>0$ and $\bm{v}\in\mathbb{R}^3$ with $0<|\bm{v}|<\Theta_{j,+}^\lambda$ the functional $\mathscr{E}_j^{\bm{v}}$ is bounded from below on $\mathcal{S}_{\lambda}$.
\end{proposition}
\begin{proof}
Let $j\in\{\mathrm{S},\mathrm{P}\}$, $\lambda>0$ and $\bm{v}\in\mathbb{R}^3\setminus\{\bm{0}\}$ be given. Choose arbitrary real functions $(\psi_0,\bm{A}_0)\in \mathcal{S}_{\lambda}$ satisfying $\|(\bm{v}\cdot\nabla)\bm{A}_0\|_{L^2}>0$ and $(\psi_0,\bm{v}\cdot\bm{A}_0\psi_0)_{L^2}>0$; if we think of $\bm{v}$ as pointing in the $x^1$-direction we can set $\bm{A}_0=(\partial_2\Xi,-\partial_1\Xi,0)$ for some standard cut-off function $\Xi\in C_0^\infty$ and let the components of $\psi_0$ be some other cut-off function with appropriate $L^2$-norm which is supported on $\bigl\{\bm{x}\in\mathbb{R}^3\,\big|\, \partial_2\Xi(\bm{x})>0\bigr\}$. We will show that if $|\bm{v}|$ is so large that the quantity $c^2\|\nabla\otimes \bm{A}_0\|_{L^2}^2-\bm{v}^2\bigl\|\bigl(\frac{\bm{v}}{|\bm{v}|}\cdot\nabla\bigr)\bm{A}_0\bigr\|_{L^2}^2$ is negative then $\mathscr{E}^{\bm{v}}_{j}$ can not be bounded from below on $\mathcal{S}_{\lambda}$. For this purpose define
\begin{align}
\psi^{\bm{v}}_R(\bm{x})=R^{-\frac{3}{2}}\mathrm{e}^{i\frac{m\bm{v}}{\hbar}\cdot \bm{x}}\psi_0\Bigl(\frac{\bm{x}}{R}\Bigr)\textrm{ and }\bm{A}^a_{R}(\bm{x})=\frac{ac}{Q}\bm{A}_0\Bigl(\frac{\bm{x}}{R}\Bigr)\label{scaledpsiA}
\end{align}
for $a,R>0$. Then $(\psi^{\bm{v}}_{R},\bm{A}^a_{R})\in\mathcal{S}_{\lambda}$ and by simply calculating each of the terms on the right hand side of \eqref{Eexp2} we get
\begin{align}
\mathscr{E}^{\bm{v}}_{j}(\psi^{\bm{v}}_R,\bm{A}^a_{R})
=\frac{\hbar^2}{2mR^2}\|\nabla\psi_0\|_{L^2}^2+\frac{a^2}{2m}\|\bm{A}_0\psi_0\|_{L^2}^2-a\bigl(\psi_0,\bm{v}\cdot\bm{A}_0\psi_0\bigr)_{L^2}-\frac{m\bm{v}^2}{2}\lambda\nonumber\\
+\frac{Ra^2}{8\pi Q^2}\Bigl(c^2\|\nabla\otimes\bm{A}_0\|_{L^2}^2-\bm{v}^2\Bigl\|\Bigl(\frac{\bm{v}}{|\bm{v}|}\cdot\nabla\Bigr)\bm{A}_0\Bigr\|_{L^2}^2\Bigr)\nonumber\\
-\frac{\hbar a1_{\{\mathrm{P}\}}(j)}{2mR}\int_{\mathbb{R}^3}\langle \psi_0(\bm{x}),\bm{\sigma}\cdot\nabla\times \bm{A}_0(\bm{x})\psi_0(\bm{x})\rangle\,\mathrm{d}\bm{x};\label{particularE}
\end{align}
Here, we explicitly use that $\psi_0$ and $\bm{A}_0$ are chosen to be real. From \eqref{particularE} we clearly see that when $|\bm{v}|$ is as described above then for any $a>0$ we have
\begin{align*}
\lim_{R\to\infty}\mathscr{E}^{\bm{v}}_{j}(\psi^{\bm{v}}_R,\bm{A}^a_{R})=-\infty
\end{align*}
and consequently $\mathscr{E}^{\bm{v}}_{j}$ is not bounded from below on $\mathcal{S}_{\lambda}$ in this case.

We now let $j\in\{\mathrm{S},\mathrm{P}\}$, $\lambda>0$ as well as $\bm{v}\in\mathbb{R}^3$ with $0<|\bm{v}|<\Theta_{j,+}^{\lambda}$ be arbitrary and consider as a first step the case where some given $(\psi,\bm{A})\in\mathcal{S}_{\lambda}$ satisfies
\begin{align}
\|\nabla\otimes\bm{A}\|_{L^2}<16\pi K_S c|Q|\lambda^{\frac{3}{4}}\frac{|\bm{v}|}{c^2-\bm{v}^2}\|\psi\|_{L^6}^{\frac{1}{2}}.\label{FirstcasePaulilow}
\end{align}
The Lichnerowicz formula \eqref{Pauli} and approximation of $\psi$ in $H^1$ by $C_0^\infty$-functions make it possible to write the quantity $\|\nabla_{j,\bm{A}+\frac{mc}{Q}\bm{v}}\psi\|_{L^2}^2$ appearing on the right hand side of \eqref{Eexp2} as $\|\nabla_{\mathrm{S},\bm{A}+\frac{mc}{Q}\bm{v}}\psi\|_{L^2}^2-1_{\{\mathrm{P}\}}(j)\frac{\hbar Q}{c}\int_{\mathbb{R}^3}\bigl\langle\psi,\bm{\sigma}\cdot\nabla\times\bm{A}\psi\bigr\rangle\,\mathrm{d}\bm{x}$. By using the diamagnetic inequality, the Hölder inequality, the Sobolev inequality and \eqref{FirstcasePaulilow} we therefore get
\begin{align*}
\|\nabla_{j,\bm{A}+\frac{mc}{Q}\bm{v}}\psi\|_{L^2}^2 &\geq \hbar^2\|\nabla|\psi|\|_{L^2}^2-1_{\{\mathrm{P}\}}(j)\frac{\hbar |Q|\lambda^{\frac{1}{4}}}{c}\|\nabla\otimes \bm{A}\|_{L^2}\|\psi\|_{L^6}^{\frac{3}{2}}\\
&\geq \frac{\hbar^2}{K_S^2}\frac{(\Theta_{j,+}^{\lambda}-|\bm{v}|)(|\bm{v}|-\Theta_{j,-}^{\lambda})}{c^2-\bm{v}^2}\|\psi\|_{L^6}^2
\end{align*}
In addition, we apply Young's inequality for products $\bigl(ab\leq \frac{a^p}{p}+\frac{b^{p'}}{p'}$ where $\frac{1}{p}+\frac{1}{p'}=1\bigr)$ and Sobolev's inequality to the term $-\frac{Q}{c}(\psi,\bm{v}\cdot\bm{A}\psi)_{L^2}$ and obtain
\begin{align}
\mathscr{E}_{j}^{\bm{v}}(\psi,\bm{A})\geq \frac{\hbar^2}{4mK_S^2}\frac{(\Theta_{j,+}^{\lambda}-|\bm{v}|)(|\bm{v}|-\Theta_{j,-}^{\lambda})}{c^2-\bm{v}^2}\|\psi\|_{L^6}^2+\frac{1}{8\pi}\Bigl(1-\frac{\bm{v}^2}{c^2}\Bigr)\|\nabla\otimes\bm{A}\|_{L^2}^2\nonumber\\
-\frac{3K_S^2|Q\bm{v}|^{\frac{4}{3}}m^{\frac{1}{3}}\lambda}{4\hbar^{\frac{2}{3}}c^{\frac{4}{3}}}\left(\frac{c^2-\bm{v}^2}{(\Theta_{j,+}^{\lambda}-|\bm{v}|)(|\bm{v}|-\Theta_{j,-}^{\lambda})}\right)^{\frac{1}{3}}\|\nabla\otimes\bm{A}\|_{L^2}^{\frac{4}{3}}-\frac{m\bm{v}^2}{2}\lambda.\label{halfway}
\end{align}
Another application of Young's inequality for products reveals that
\begin{align}
\mathscr{E}_{j}^{\bm{v}}(\psi,\bm{A})\geq \frac{\hbar^2}{4mK_S^2}\frac{(\Theta_{j,+}^{\lambda}-|\bm{v}|)(|\bm{v}|-\Theta_{j,-}^{\lambda})}{c^2-\bm{v}^2}\|\psi\|_{L^6}^2+\frac{1}{16\pi}\Bigl(1-\frac{\bm{v}^2}{c^2}\Bigr)\|\nabla\otimes\bm{A}\|_{L^2}^2\nonumber\\
-\frac{(4\pi)^2K_S^6Q^4m\lambda^3}{\hbar^2}\frac{\bm{v}^4}{(c^2-\bm{v}^2)(\Theta_{j,+}^{\lambda}-|\bm{v}|)(|\bm{v}|-\Theta_{j,-}^{\lambda})}-\frac{m\bm{v}^2}{2}\lambda\label{estimateonenergyFCase}
\end{align}
so for pairs $(\psi,\bm{A})\in\mathcal{S}_{\lambda}$ satisfying \eqref{FirstcasePaulilow} there is indeed a lower bound on the possible values of $\mathscr{E}_j^{\bm{v}}(\psi,\bm{A})$. Consider now the scenario where the given pair $(\psi,\bm{A})\in \mathcal{S}_{\lambda}$ satisfies the inequality
\begin{align}
\|\nabla\otimes\bm{A}\|_{L^2}\geq16\pi K_S c|Q|\lambda^{\frac{3}{4}}\frac{|\bm{v}|}{c^2-\bm{v}^2}\|\psi\|_{L^6}^{\frac{1}{2}}.\label{SecondPaulilow}
\end{align}
In this case we simply use the nonnegativity of the kinetic energy term in \eqref{Eexp2} to get
\begin{align}
\mathscr{E}_{j}^{\bm{v}}(\psi,\bm{A})&\!\geq \!-\frac{K_S|Q||\bm{v}|\lambda^{\frac{3}{4}}}{c}\|\nabla\otimes\bm{A}\|_{L^2}\|\psi\|_{L^6}^{\frac{1}{2}}\!-\frac{m\bm{v}^2}{2}\lambda+\!\frac{1}{8\pi}\Bigl(1-\frac{\bm{v}^2}{c^2}\Bigr)\|\nabla\otimes\bm{A}\|_{L^2}^2\nonumber\\
&\geq \frac{1}{16\pi}\Bigl(1-\frac{\bm{v}^2}{c^2}\Bigr)\|\nabla\otimes\bm{A}\|_{L^2}^2-\frac{m\bm{v}^2}{2}\lambda\label{estimateonenergySCase}
\end{align}
where the assumption \eqref{SecondPaulilow} is applied at the final step. Consequently, the values of $\mathscr{E}_j^{\bm{v}}(\psi,\bm{A})$ are bounded below by $-\frac{m\bm{v}^2}{2}$ for pairs $(\psi,\bm{A})\in\mathcal{S}_{\lambda}$ satisfying \eqref{SecondPaulilow}.
\end{proof}
\begin{bemaerkning}
For $j\in\{\mathrm{S},\mathrm{P}\}$, $\lambda>0$, $\bm{v}\in \mathbb{R}^3$ with $0<|\bm{v}|<\Theta_{j,+}^{\lambda}$ as well as $(\psi,\bm{A})\in\mathcal{S}_{\lambda}$ we can bound the quantities $\|\psi\|_{L^6}$ and $\|\nabla\otimes\bm{A}\|_{L^2}$ from above in terms of $\mathscr{E}_j^{\bm{v}}(\psi,\bm{A})$. More precisely, \eqref{estimateonenergyFCase} and \eqref{estimateonenergySCase} give that
\begin{align}
\|\nabla\otimes\bm{A}\|_{L^2}^2\leq \frac{2^8\pi^3K_S^6c^2Q^4m\lambda^3}{\hbar^2}\frac{\bm{v}^4}{(c^2-\bm{v}^2)^2(\Theta_{j,+}^{\lambda}-|\bm{v}|)(|\bm{v}|-\Theta_{j,-}^{\lambda})}\nonumber\\
+\frac{16\pi c^2}{c^2-\bm{v}^2}\Bigl(\mathscr{E}_j^{\bm{v}}(\psi,\bm{A})+\frac{m\bm{v}^2}{2}\lambda\Bigr).\label{boundonA}
\end{align}
Moreover, if $(\psi,\bm{A})$ satisfies \eqref{FirstcasePaulilow} we obtain from \eqref{estimateonenergyFCase} that
\begin{align}
\|\psi\|_{L^6}^2\leq \frac{4mK_S^2}{\hbar^2}\frac{c^2-\bm{v}^2}{(\Theta_{j,+}^{\lambda}-|\bm{v}|)(|\bm{v}|-\Theta_{j,-}^{\lambda})}\Bigl(\mathscr{E}_j^{\bm{v}}(\psi,\bm{A})+\frac{m\bm{v}^2}{2}\lambda\Bigr)\nonumber\\
+\frac{2^6\pi^2K_S^8Q^4m^2\lambda^3}{\hbar^4}\frac{\bm{v}^4}{(\Theta_{j,+}^{\lambda}-|\bm{v}|)^2(|\bm{v}|-\Theta_{j,-}^{\lambda})^2}\label{boundonpsifirst}
\end{align}
and if $(\psi,\bm{A})$ on the other hand satisfies \eqref{SecondPaulilow} then by \eqref{estimateonenergySCase} we have
\begin{align}
\|\psi\|_{L^6}\leq \frac{c^2-\bm{v}^2}{16\pi K_S^2Q^2\lambda^{\frac{3}{2}}\bm{v}^2}\Bigl(\mathscr{E}_j^{\bm{v}}(\psi,\bm{A})+\frac{m\bm{v}^2}{2}\lambda\Bigr).\label{boundonpsisecond}
\end{align}
\end{bemaerkning}
By Proposition \ref{lowerb} it is impossible for the functional $\mathscr{E}_j^{\bm{v}}$ to attain a minimum on $\mathcal{S}_{\lambda}$ for sufficiently large values of $|\bm{v}|$. Of course this does not rule out the existence of solutions to \eqref{internMS2}, but the nonexistence of such solutions for large $|\bm{v}|$ would in fact be perfectly compatible with our understanding from the theory of special relativity that a particle with rest mass can not travel faster than light. We therefore guess that the value $\Theta_{\mathrm{S},+}^{\lambda}$ is optimal in the sense that $\mathscr{E}_{\mathrm{S}}^{\bm{v}}$ can not be shown to be bounded from below on $\mathcal{S}_{\lambda}$ for $|\bm{v}|>c$. On the other hand, the value of $\Theta_{\mathrm{P},+}^{\lambda}$ is not optimal.

\chapter{Properties of the infimum}\label{propertiesofIsection}
For any given $\bm{v}\in\mathbb{R}^3$ consider the set
\begin{align*}
\Lambda_j^{\bm{v}}=\bigl\{\lambda>0\,\mid\, |\bm{v}|<\Theta_{j,+}^\lambda\bigr\}.
\end{align*}
\begin{figure}[ht]
	\centering
		\includegraphics[width=0.45\textwidth]{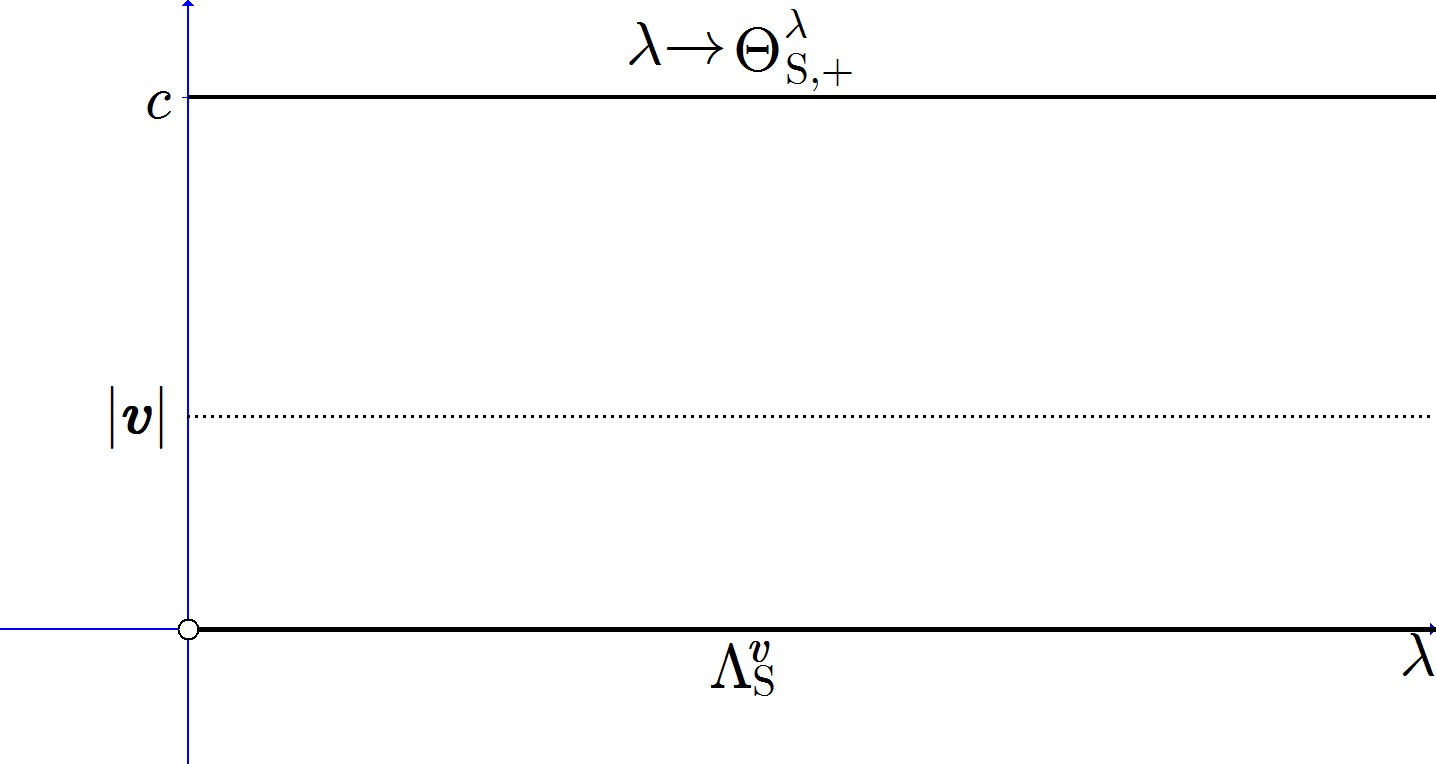}\hspace{1cm}
		\includegraphics[width=0.45\textwidth]{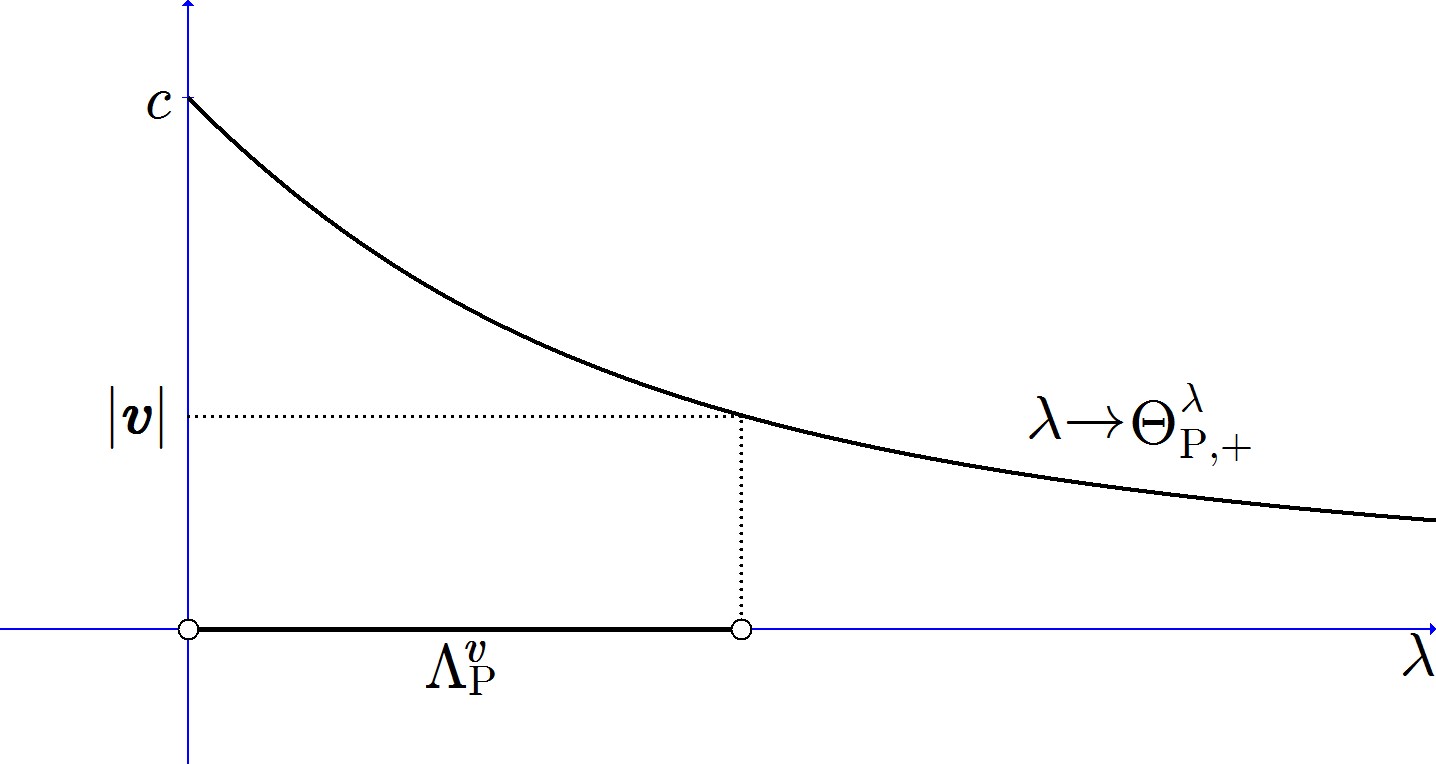}
	\caption{The set $\Lambda_j^{\bm{v}}$ is an open interval since $\lambda\mapsto\Theta_{j,+}^\lambda$ is decreasing and continuous. In fact, $\Lambda_{\mathrm{S}}^{\bm{v}}=(0,\infty)$ regardless of the choice of $\bm{v}$ with $|\bm{v}|<c$.}
	\label{fig:ThetaS}
\end{figure}
We have just seen that given $\bm{v}\in\mathbb{R}^3$ with $0<|\bm{v}|<c$ and $\lambda\in \Lambda_j^{\bm{v}}$ it makes sense to define
\begin{align*}
I^{\lambda}_j=\inf\bigl\{\mathscr{E}_j^{\bm{v}}(\psi,\bm{A})\,\big|\, (\psi,\bm{A})\in \mathcal{S}_{\lambda}\bigr\}
\end{align*}
and we aim to show that this infimum is attained. Imagine that the functional $\mathscr{E}_j^{\bm{v}}$ indeed does take the value $I_j^\lambda$ in some point. It follows from the following simple observation that for such a minimizing point $(\psi,\bm{A})\in \mathcal{S}_{\lambda}$ neither the wave function $\psi$ nor the magnetic vector potential $\bm{A}$ can be identically equal to zero.
\begin{lemma}\label{negativeI}
Let $j\in\{\mathrm{S},\mathrm{P}\}$ and $\bm{v}\in\mathbb{R}^3$ with $0<|\bm{v}|<c$ be given. Then
\begin{align*}
I_j^\lambda<-\frac{m\bm{v}^2}{2}\lambda
\end{align*}
for any $\lambda\in \Lambda_j^{\bm{v}}$.
\end{lemma}
\begin{proof}
Choose the pair $(\psi_0,\bm{A}_0)\in \mathcal{S}_{\lambda}$ as in the beginning of the proof of Proposition \ref{lowerb} and define $(\psi_R^{\bm{v}},\bm{A}_R^a)\in \mathcal{S}_{\lambda}$ for $R,a>0$ as prescribed in \eqref{scaledpsiA}. According to \eqref{particularE} we can let $R$ take the specific value
\begin{align*}
R_a=\left(\frac{2\hbar Q}{a}\sqrt{\frac{\pi}{m}}\|\nabla\psi_0\|_{L^2}\right)^{\frac{2}{3}}\bigl(c^2\|\nabla\otimes\bm{A}_0\|_{L^2}^2-\|(\bm{v}\cdot\nabla)\bm{A}_0\|_{L^2}^2\bigr)^{-\frac{1}{3}},
\end{align*}
and get
\begin{align*}
\mathscr{E}^{\bm{v}}_{j}(\psi^{\bm{v}}_{R_a},\bm{A}^a_{R_a})
=\Bigl(\tfrac{\hbar^2}{16\pi^2 Q^4m}\Bigr)^{\frac{1}{3}}\bigl(c^2\|\nabla\otimes\bm{A}_0\|_{L^2}^2-\|(\bm{v}\cdot\nabla)\bm{A}_0\|_{L^2}^2\bigr)^{\frac{2}{3}}\|\nabla\psi_0\|_{L^2}^{\frac{2}{3}}a^{\frac{4}{3}}\\
-\Bigl(\tfrac{1_{\{\mathrm{P}\}}(j)\hbar(c^2\|\nabla\otimes\bm{A}_0\|_{L^2}^2-\|(\bm{v}\cdot\nabla)\bm{A}_0\|_{L^2}^2)}{2^5m^2Q^2\pi\|\nabla\psi_0\|_{L^2}^2}\Bigr)^{\frac{1}{3}}\!\!\!\int_{\mathbb{R}^3}\!\!\!\langle \psi_0,\bm{\sigma}\cdot\nabla\times \bm{A}_0\psi_0\rangle\mathrm{d}\bm{x} a^{\frac{5}{3}}\\
+\tfrac{1}{2m}\|\bm{A}_0\psi_0\|_{L^2}^2a^2-(\psi_0,\bm{v}\cdot\bm{A}_0\psi_0)_{L^2}a-\tfrac{m\bm{v}^2}{2}\lambda.
\end{align*}
Thus, $\mathscr{E}^{\bm{v}}_j(\psi^{\bm{v}}_{R_a},\bm{A}^a_{R_a})$ can be extended to a continuously differentiable function of $a$ on the entire real line -- moreover, the extension takes the value $-\frac{m\bm{v}^2}{2}\lambda$ and has a negative derivative at $a=0$. For sufficiently small $a>0$ the values of $\mathscr{E}^{\bm{v}}_j(\psi^{\bm{v}}_{R_a},\bm{A}^a_{R_a})$ must therefore be strictly less than $-\frac{m\bm{v}^2}{2}\lambda$.
\end{proof}
In the following proposition we investigate $I_j^\lambda$'s dependence on $\lambda$.
\begin{lemma}
Let $j\in\{\mathrm{S},\mathrm{P}\}$ and $\bm{v}\in\mathbb{R}^3$ with $0<|\bm{v}|<c$ be given. Then
\begin{align}
I_j^{s\nu}<sI_j^{\nu}\label{contI}
\end{align}
for all $\nu\in \Lambda_j^{\bm{v}}$ and $s>1$ with $s\nu\in\Lambda_j^{\bm{v}}$. Moreover,
\begin{align}
I_j^{\lambda}< I_j^{\mu}+ I_j^{\lambda-\mu}\label{subadditivity}
\end{align}
for $\mu,\lambda\in\Lambda_j^{\bm{v}}$ with $\mu<\lambda$.
\end{lemma}
\begin{proof}
Let $\nu\in\Lambda_{j}^{\bm{v}}$ and $s>1$ with $s\nu\in\Lambda_{j}^{\bm{v}}$ be given and choose (by means of Lemma \ref{negativeI}) some constant $k\in \bigl(0,-\frac{m\bm{v}^2}{2}\nu-I_{j}^{\nu}\bigr)$. Given a positive $\varepsilon$ satisfying
\begin{align*}
\varepsilon<\min\left\{-\frac{m\bm{v}^2}{2}\nu-I^\nu_{j}-k,\frac{s-1}{s}\frac{k^{\frac{3}{2}}\hbar\sqrt{\bigl(c^2-\bm{v}^2\bigr)\bigl(\Theta_{j,+}^\nu-|\bm{v}|\bigr)\bigl(|\bm{v}|-\Theta_{j,-}^\nu\bigr)}}{3^{\frac{3}{2}}\pi K_S^3Q^2\sqrt{m}\nu^{\frac{3}{2}}\bm{v}^2}\right\}
\end{align*}
we can then choose a pair $(\psi,\bm{A})\in \mathcal{S}_{\nu}$ such that
\begin{align}
\mathscr{E}_{j}^{\bm{v}}(\psi,\bm{A})\leq I^\nu_{j}+\varepsilon,\label{EpsiA}
\end{align}
which together with \eqref{halfway}, \eqref{estimateonenergySCase} and the assumption $\varepsilon<-\frac{m\bm{v}^2}{2}\nu-I^\nu_{j}-k$ gives
\begin{align}
\|\nabla\otimes \bm{A}\|_{L^2}^2>\frac{8k^{\frac{3}{2}}\hbar c^2}{3^{\frac{3}{2}}K_S^3 Q^2\sqrt{m}\nu^{\frac{3}{2}}\bm{v}^2}\left(\frac{\bigl(\Theta_{j,+}^\nu-|\bm{v}|\bigr)\bigl(|\bm{v}|-\Theta_{j,-}^\nu\bigr)}{c^2-\bm{v}^2}\right)^{\frac{1}{2}}.\label{lowerboundonA}
\end{align}
Then \eqref{EpsiA} and \eqref{lowerboundonA} imply that
\begin{align}
I^{s\nu}_{j}&\leq \mathscr{E}_{j}^{\bm{v}}(\sqrt{s}\psi,\bm{A})\nonumber\\
&= s\mathscr{E}_{j}^{\bm{v}}(\psi,\bm{A})+\frac{1-s}{8\pi}\Bigl(\|\nabla\otimes\bm{A}\|_{L^2}^2-\Bigl\|\Bigl(\frac{\bm{v}}{c}\cdot\nabla\Bigr)\bm{A}\Bigr\|_{L^2}^2\Bigr)\label{ssqrt}\\
&< s I^\nu_{j} + s\varepsilon + (1-s)\frac{k^{\frac{3}{2}}\hbar\sqrt{\bigl(c^2-\bm{v}^2\bigr)\bigl(\Theta_{j,+}^\nu-|\bm{v}|\bigr)\bigl(|\bm{v}|-\Theta_{j,-}^\nu\bigr)}}{3^{\frac{3}{2}}\pi K_S^3Q^2\sqrt{m}\nu^{\frac{3}{2}}\bm{v}^2}\nonumber\\
&<s I^\nu_{j},\nonumber
\end{align}
proving that \eqref{contI} indeed does hold true.

This enables us to prove \eqref{subadditivity}, so let $\mu,\lambda\in\Lambda_j^{\bm{v}}$ with $\mu<\lambda$ be given. If $\mu>\lambda-\mu$ is satisfied we can use \eqref{contI} twice $\bigl($with $(s,\nu)=\bigl(\frac{\lambda}{\mu},\mu\bigr)$ respectively $(s,\nu)=\bigl(\frac{\mu}{\lambda-\mu},\lambda-\mu\bigr)\bigr)$ and obtain
\begin{align*}
I^\lambda_{j}=I^{\frac{\lambda}{\mu}\mu}_{j}<\frac{\lambda}{\mu}I^\mu_{j}=I^\mu_{j}+\frac{\lambda-\mu}{\mu}I^{\frac{\mu}{\lambda-\mu}(\lambda-\mu)}_{j}< I^\mu_{j}+I^{\lambda-\mu}_{j}
\end{align*}
and if on the other hand $\mu\leq\lambda-\mu$ we can likewise apply \eqref{contI} to get
\begin{align*}
I^\lambda_j= I^{\frac{\lambda}{\lambda-\mu}(\lambda-\mu)}_j<\frac{\lambda}{\lambda-\mu}I^{\lambda-\mu}_j=\frac{\mu}{\lambda-\mu}I^{\frac{\lambda-\mu}{\mu}\mu}_j+I^{\lambda-\mu}_j\leq I^{\mu}_j+I^{\lambda-\mu}_j,
\end{align*}
so \eqref{subadditivity} also holds true.
\end{proof}
\begin{bemaerkning}\label{Idecreasing}
The strict subadditivity expressed in \eqref{subadditivity} implies that $\nu\mapsto I_j^\nu$ is strictly decreasing on $\Lambda_j^{\bm{v}}$ since the term $I_j^{\lambda-\mu}$ is negative by Lemma \ref{negativeI}.
\end{bemaerkning}
As a consequence of Remark \ref{Idecreasing} the function $\nu\mapsto I_j^\nu$ has limits from the left as well as from the right in all points of $\Lambda_j^{\bm{v}}$. In fact, we can show the following result.
\begin{lemma}\label{Icontinuous}
Given $j\in\{\mathrm{S},\mathrm{P}\}$ and $\bm{v}\in\mathbb{R}^3$ with $0<|\bm{v}|<c$ the mapping $\nu\mapsto I_j^{\nu}$ is continuous on $\Lambda_j^{\bm{v}}$.
\end{lemma}
\begin{proof}
Let us begin by proving that $\nu\mapsto I_j^{\nu}$ is left continuous: Given $\nu\in \Lambda_j^{\bm{v}}$, $\varepsilon>0$ and $0<s<1$ we choose $(\psi,\bm{A})\in\mathcal{S}_{\nu}$ such that $\mathscr{E}_j^{\bm{v}}(\psi,\bm{A})\leq I_j^\nu+\varepsilon$ and proceed just as in \eqref{ssqrt} to obtain
\begin{align*}
I_j^{s\nu}&\leq sI_j^\nu+s\varepsilon+\frac{1-s}{8\pi}\Bigl(\|\nabla\otimes\bm{A}\|_{L^2}^2-\Bigl\|\Bigl(\frac{\bm{v}}{c}\cdot\nabla\Bigr)\bm{A}\Bigr\|_{L^2}^2\Bigr).
\end{align*}
Letting $s\to 1^-$ therefore gives $\lim_{s\to 1^-}I_j^{s\nu}\leq I_j^\nu+\varepsilon$ and the fact that this holds true for any $\varepsilon>0$ implies that $\lim_{s\to 1^-}I_j^{s\nu}\leq I_j^\nu$. Since $\nu\mapsto I_j^\nu$ is decreasing the opposite inequality also holds true, whereby
\begin{align*}
\lim_{s\to 1^-}I_j^{s\nu}=I_j^\nu.
\end{align*}

To prove right continuity of $\nu\mapsto I_j^\nu$ we let $\nu\in \Lambda_j^{\bm{v}}$, $\varepsilon>0$ as well as $s>1$ with $s\nu\in \Lambda_j^{\bm{v}}$ be arbitrary and choose a pair $(\psi,\bm{A})\in \mathcal{S}_{s\nu}$ such that $\mathscr{E}_j^{\bm{v}}(\psi,\bm{A})\leq I_j^{s\nu}+\varepsilon$. Then \eqref{boundonA} and Lemma \ref{negativeI} give that
\begin{align*}
I_j^{s\nu}+\varepsilon&\geq s\mathscr{E}_j^{\bm{v}}\bigl(\tfrac{\psi}{\sqrt{s}},\bm{A}\bigr)+\frac{1-s}{8\pi}\Bigl(\|\nabla\otimes\bm{A}\|_{L^2}^2-\Bigl\|\Bigl(\frac{\bm{v}}{c}\cdot\nabla\Bigr)\bm{A}\Bigr\|_{L^2}^2\Bigr)\\
&\geq sI_j^\nu+\frac{2c^2(1-s)}{c^2-\bm{v}^2}\left(\varepsilon+\frac{(4\pi)^2K_S^6Q^4m(s\nu)^3\bm{v}^4}{\hbar^2(c^2-\bm{v}^2)(\Theta_{j,+}^{s\nu}-|\bm{v}|)(|\bm{v}|-\Theta_{j,-}^{s\nu})}\right),
\end{align*}
whereby $\lim_{s\to 1^+}I_j^{s\nu}+\varepsilon\geq I_j^\nu$. By letting $\varepsilon\to 0^+$ we thus obtain the inequality $\lim_{s\to 1^+}I_j^{s\nu}\geq I_j^\nu$ and the opposite inequality follows immediately from \eqref{contI}, which leaves us in position to conclude that the identity
\begin{align*}
\lim_{s\to 1^+}I_j^{s\nu}=I_j^\nu
\end{align*}
holds true.
\end{proof}

\chapter{Existence of a Minimizer}\label{existenceofminim}
We will now consider a strategy that is frequently used for approaching minimization problems such as ours -- it is often called the direct method in the calculus of variations and was introduced by Zaremba and Hilbert around the year 1900. Here, one first considers a minimizing sequence for the functional at hand.
\begin{definition}
Let $j\in\{\mathrm{S},\mathrm{P}\}$, $\bm{v}\in\mathbb{R}^3$ with $0<|\bm{v}|<c$ and $\lambda\in \Lambda_j^{\bm{v}}$ be given. By a \emph{minimizing sequence} for $\mathscr{E}_j^{\bm{v}}$ we mean a sequence of points $(\psi_n,\bm{A}_n)\in\mathcal{S}_{\lambda}$ such that $\bigl(\mathscr{E}_j^{\bm{v}}(\psi_n,\bm{A}_n)\bigr)_{n\in\mathbb{N}}$ converges to $I_j^\lambda$ in $\mathbb{R}$.
\end{definition}
The philosophy of the direct method in the calculus of variations is to first argue that a given minimizing sequence must have a subsequence converging weakly to some point $(\psi,\bm{A})$ and then as a second step one hopes to show lower semicontinuity properties of $\mathscr{E}_j^{\bm{v}}$ ensuring that the identity $\mathscr{E}_j^{\bm{v}}(\psi,\bm{A})=I_j^\lambda$ holds true. However, our specific functional $\mathscr{E}_j^{\bm{v}}$ is translation invariant -- meaning that any translation $\tau_{\bm{y}}:\bm{x}\mapsto(\bm{x}+\bm{y})$ in space gives rise to the identity $\mathscr{E}_j^{\bm{v}}(\psi\circ\tau_{\bm{y}},\bm{A}\circ\tau_{\bm{y}})=\mathscr{E}_j^{\bm{v}}(\psi,\bm{A})$. Thus, even if $\mathscr{E}_j^{\bm{v}}$ indeed does have a minimizer, there will exist lots of minimizing sequences whose $\psi$-part converges weakly in $L^2$ to the zero function and the possible limit of (any subsequence of) such a minimizing sequence can clearly not serve as a minimizer for $\mathscr{E}_j^{\bm{v}}$. In other words, we have to break the translation invariance in some way and to do this we will prove a variant of the concentration-compactness principle by Pierre-Louis Lions (see \cite{Lions1,Lions2}). We can not just apply the result of Lions to our problem since this result concerns a sequence of $H^1$- (or $L^1$-)functions $\psi_n$ whereas we are dealing with a sequence of $\mathcal{S}_{\lambda}$-pairs $(\psi_n,\bm{A}_n)$. Let us begin by proving the following simple -- but important -- lemma that provides us with some control over any given minimizing sequence for $\mathscr{E}_j^{\bm{v}}$.
\begin{lemma}\label{minbdd}
Let $j\in\{\mathrm{S},\mathrm{P}\}$, $\bm{v}\in\mathbb{R}^3$ with $0<|\bm{v}|<c$ as well as $\lambda\in\Lambda_j^{\bm{v}}$ be given and consider a minimizing sequence $\bigl((\psi_n,\bm{A}_n)\bigr)_{n\in\mathbb{N}}\subset \mathcal{S}_{\lambda}$ for $\mathscr{E}_j^{\bm{v}}$. Then $(\psi_n)_{n\in\mathbb{N}}$ is bounded in $H^1$ and $(\bm{A}_n)_{n\in\mathbb{N}}$ is bounded in $D^1$.
\end{lemma}
\begin{proof}
The sequence $\bigl(\mathscr{E}_j^{\bm{v}}(\psi_n,\bm{A}_n)\bigr)_{n\in\mathbb{N}}$ is bounded (because it is convergent) and therefore it follows from the estimates \eqref{boundonA}, \eqref{boundonpsifirst} and \eqref{boundonpsisecond} that $(\|\psi_n\|_{L^6})_{n\in\mathbb{N}}$ and $(\|\nabla\otimes\bm{A}_n\|_{L^2})_{n\in\mathbb{N}}$ are also bounded. Moreover, the sequence $(\|\psi_n\|_{L^2})_{n\in\mathbb{N}}$ is constant so all that remains to be shown is the boundedness of $(\|\nabla\psi_n\|_{L^2})_{n\in\mathbb{N}}$. For this we expand the kinetic energy in the expression \eqref{Eexp1} for $\mathscr{E}_j^{\bm{v}}(\psi_n,\bm{A}_n)$, use the nonnegativity of $\frac{Q^2}{2mc^2}\|\bm{A}_n\psi_n\|_{L^2}^2+\frac{1}{8\pi}\bigl(\|\nabla\otimes\bm{A}_n\|_{L^2}^2-\bigl\|\bigl(\frac{\bm{v}}{c}\cdot\nabla\bigr)\bm{A}_n\bigr\|^2\bigr)$ and apply Hölder's as well as Sobolev's inequalities to get
\begin{align}
\frac{\hbar^2}{2m}\|\nabla\psi_n\|_{L^2}^2\leq |\mathscr{E}_j^{\bm{v}}(\psi_n,\bm{A}_n)|+\frac{K_S\hbar |Q|\lambda^{\frac{1}{4}}}{mc}\|\nabla\psi_n\|_{L^2}\|\nabla\otimes\bm{A}_n\|_{L^2}\|\psi_n\|_{L^6}^{\frac{1}{2}}\nonumber\\
+\hbar\lambda^{\frac{1}{2}}|\bm{v}|\|\nabla\psi_n\|_{L^2}+1_{\{\mathrm{P}\}}(j)\frac{\hbar |Q|\lambda^{\frac{1}{4}}}{2mc}\|\nabla\otimes\bm{A}_n\|_{L^2}\|\psi_n\|_{L^6}^{\frac{3}{2}}.\label{boundonnablapsi}
\end{align}
Here, we can use Young's inequality for products to absorb the $\|\nabla\psi_n\|_{L^2}$'s on the right hand side of \eqref{boundonnablapsi} into the left hand side of \eqref{boundonnablapsi} and obtain an upper bound on $\|\nabla\psi_n\|_{L^2}^2$. 
\end{proof}
\begin{bemaerkning}\label{bdn}
From now on we will consider some fixed minimizing sequence $\bigl((\psi_n,\bm{A}_n)\bigr)_{n\in\mathbb{N}}$ for $\mathscr{E}_j^{\bm{v}}$. It follows from Lemma \ref{minbdd} that all of the terms appearing in the expressions \eqref{Eexp1} and \eqref{Eexp2} for $\mathscr{E}_j^{\bm{v}}(\psi_n,\bm{A}_n)$ define bounded sequences in $\mathbb{R}$. We will let $C$ denote a constant that majorizes each of the sequences $(\|\psi_n\|_{H^1})_{n\in\mathbb{N}}$, $(\|\nabla\otimes\bm{A}_n\|_{L^2})_{n\in\mathbb{N}}$ and $(\|\nabla_{j,\bm{A}_n+\frac{mc}{Q}\bm{v}}\psi_n\|_{L^2})_{n\in\mathbb{N}}$.
\end{bemaerkning}

\section{Breaking the Translation Invariance}
We hope to find a sequence $(\bm{y}_n)_{n\in\mathbb{N}}$ of points in $\mathbb{R}^3$ such that the direct method in the calculus of variations can be applied to the translated minimizing sequence $\bigl((\psi_n\circ \tau_{\bm{y}_n},\bm{A}_n\circ\tau_{\bm{y}_n})\bigr)_{n\in\mathbb{N}}$ for $\mathscr{E}_j^{\bm{v}}$. As an essential tool in our search for such a sequence we introduce for each $n\in\mathbb{N}$ the nondecreasing \emph{concentration function} $\mathscr{C}_n:(0,\infty)\to (0,\lambda]$ given by
\begin{align}
\mathscr{C}_n(r)=\sup_{\bm{y}\in\mathbb{R}^3}\int_{\mathcal{B}(\bm{y},r)}|\psi_n(\bm{x})|^2\,\mathrm{d}\bm{x}\quad\textrm{for}\quad r>0.\label{Qndef}
\end{align}
Remember that we think of the $\psi$-variable as being a quantum particle's wave function and so the physical interpretation of a large value of $\mathscr{C}_n(r)$ (compared to $\lambda$) is that the quantum particle is likely to be localized in some ball $\subset \mathbb{R}^3$ with radius $r$. In this sense $\mathscr{C}_n$ expresses how concentrated the wave function is (see Figure \ref{fig:Concentration}).
\begin{figure}[ht]
	\centering
		\includegraphics[width=0.6\textwidth]{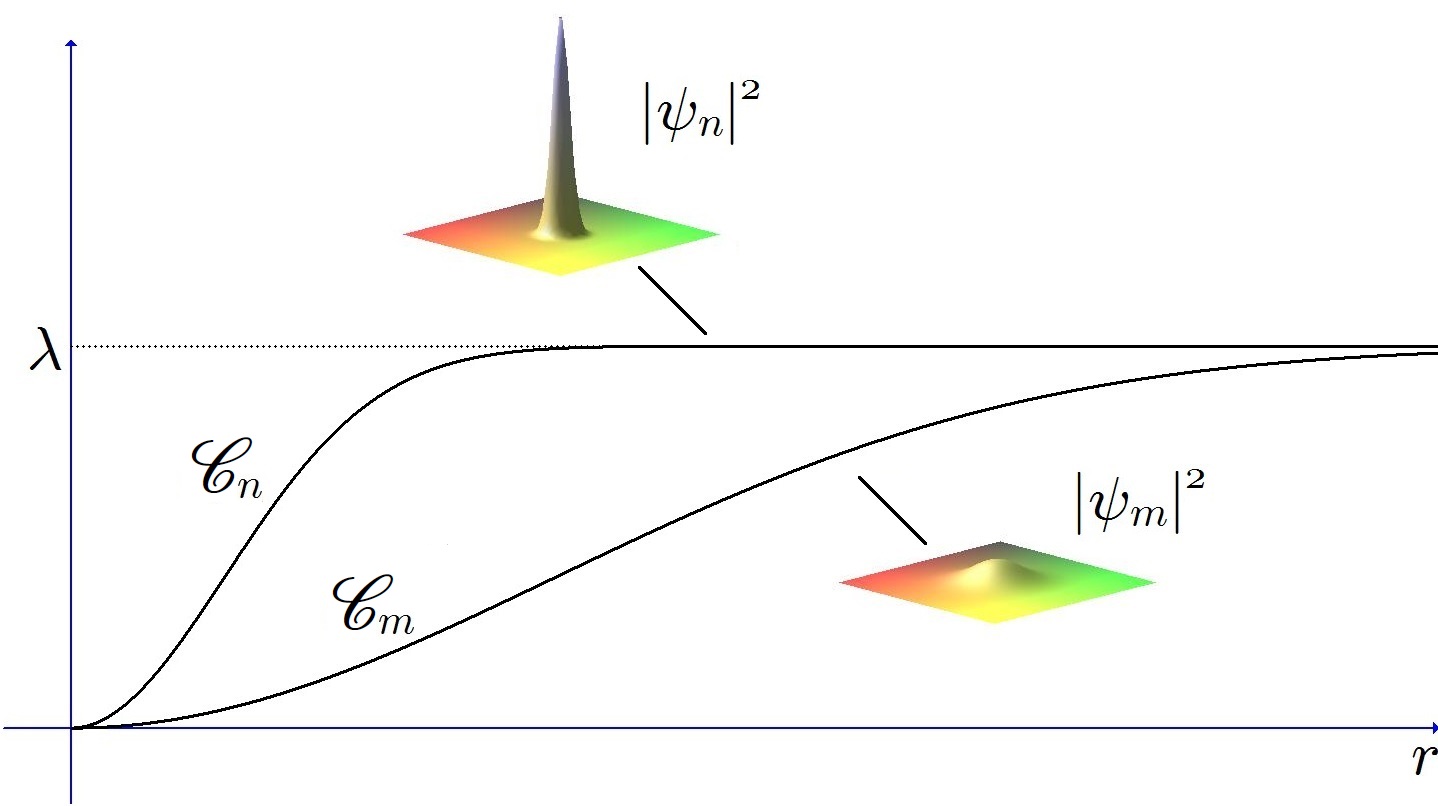}
	\caption{If $\mathscr{C}_n$ increases quickly to the value $\lambda$ then it means that the corresponding wave function $\psi_n$ is very concentrated around some point $\bm{y}$ in space, i.e. the quantum particle is with high probability positioned in close vicinity of $\bm{y}$.}
	\label{fig:Concentration}
\end{figure}
We summarize the most important properties of the functions $\mathscr{C}_n$ in the following lemma.
\begin{lemma}\label{Qnlemma}
Given $\lambda>0$ let $\psi_n\in H^1$ satisfy $\|\psi_n\|_{L^2}^2=\lambda$ and $\|\nabla\psi_n\|_{L^2}\leq C$ for all $n\in\mathbb{N}$ and define the function $\mathscr{C}_n:(0,\infty)\to (0,\lambda]$ by \eqref{Qndef}. Then $\mathscr{C}_n$ is nondecreasing with the limits $\lim_{r\to 0^+}\mathscr{C}_n(r)=0$ as well as $\lim_{r\to \infty}\mathscr{C}_n(r)=\lambda$ holding true and by passing to a subsequence $(\mathscr{C}_n)_{n\in\mathbb{N}}$ converges pointwise to some nondecreasing mapping $\mathscr{C}:(0,\infty)\to [0,\lambda]$ with $\lim_{r\to 0^+}\mathscr{C}(r)=0$.
\end{lemma}
\begin{proof}
For an arbitrary $n\in\mathbb{N}$ the mapping $\mathscr{C}_n:(0,\infty)\to (0,\lambda]$ is obviously nondecreasing and the identity $\lim_{r\to 0^+}\mathscr{C}_n(r)=0$ holds true since
\begin{align}
\int_{\mathcal{B}(\bm{y},r)}|\psi_n(\bm{x})|^2\,\mathrm{d}\bm{x}\leq \Bigl(\frac{4}{3}\pi r^3K_S^3C^3\Bigr)^{\frac{2}{3}}\label{uniformQ}
\end{align}
for all $(r,\bm{y})\in(0,\infty)\times\mathbb{R}^3$ by Hölder's and Sobolev's inequalities. Moreover, we have $\lim_{r\to \infty}\mathscr{C}_n(r)=\lambda$ since Lebesgue's theorem on dominated convergence gives that the difference
\begin{align*}
\lambda-\mathscr{C}_n(r)\leq \lambda-\int_{\mathcal{B}(\bm{0},r)}|\psi_n(\bm{x})|^2\,\mathrm{d}\bm{x}=\int_{\mathbb{R}^3\setminus\mathcal{B}(\bm{0},r)}|\psi_n(\bm{x})|^2\,\mathrm{d}\bm{x}
\end{align*}
can be made arbitrarily small by choosing $r$ sufficiently large. Helly's selection principle \cite[Theorem 10.5]{KolF} ensures the existence of a subsequence of $(\mathscr{C}_n)_{n\in\mathbb{N}}$ converging pointwise to some function $\mathscr{C}$. The limit function $\mathscr{C}$ inherits the nondecreasingness from the $\mathscr{C}_n$'s and \eqref{uniformQ} gives that $\lim_{r\to 0^+}\mathscr{C}(r)=0$.
\end{proof}
To simplify notation we will also denote the subsequence described in Lemma \ref{Qnlemma} by $(\mathscr{C}_n)_{n\in\mathbb{N}}$. It is apparent that $\mathscr{C}$ and the $\mathscr{C}_n$-functions have almost identical properties. But even though the lemma depicts $\lim_{r\to\infty}\mathscr{C}_n(r)$ as being equal to $\lambda$ it does not at all mention the value of the limit
\begin{align}
\mu:=\lim_{r\to \infty}\mathscr{C}(r),\label{alphadef}
\end{align}
which is obviously well defined and contained in the interval $[0,\lambda]$. To determine the value of $\mu$ we first turn to our physical intuition: Remember that the points $(\psi_n,\bm{A}_n)$ form a minimizing sequence and we hope to show weak convergence (in some sense) of these points to a pair $(\psi,\bm{A})$ minimizing $\mathscr{E}_j^{\bm{v}}$. For a moment let us focus on the $\psi$-variable: It can be fruitful to think of our quantum particle as being prepared in some initial state and as time evolves we receive snapshots (corresponding to the sequence elements $\psi_1,\psi_2,\psi_3,\ldots$) of the system's intermediate states that steadily approach the limiting state $\psi$, which has the least possible energy. 
\begin{figure}[ht]
\centering
\subbottom[$\mu=0$]{\includegraphics[width=0.3\textwidth]{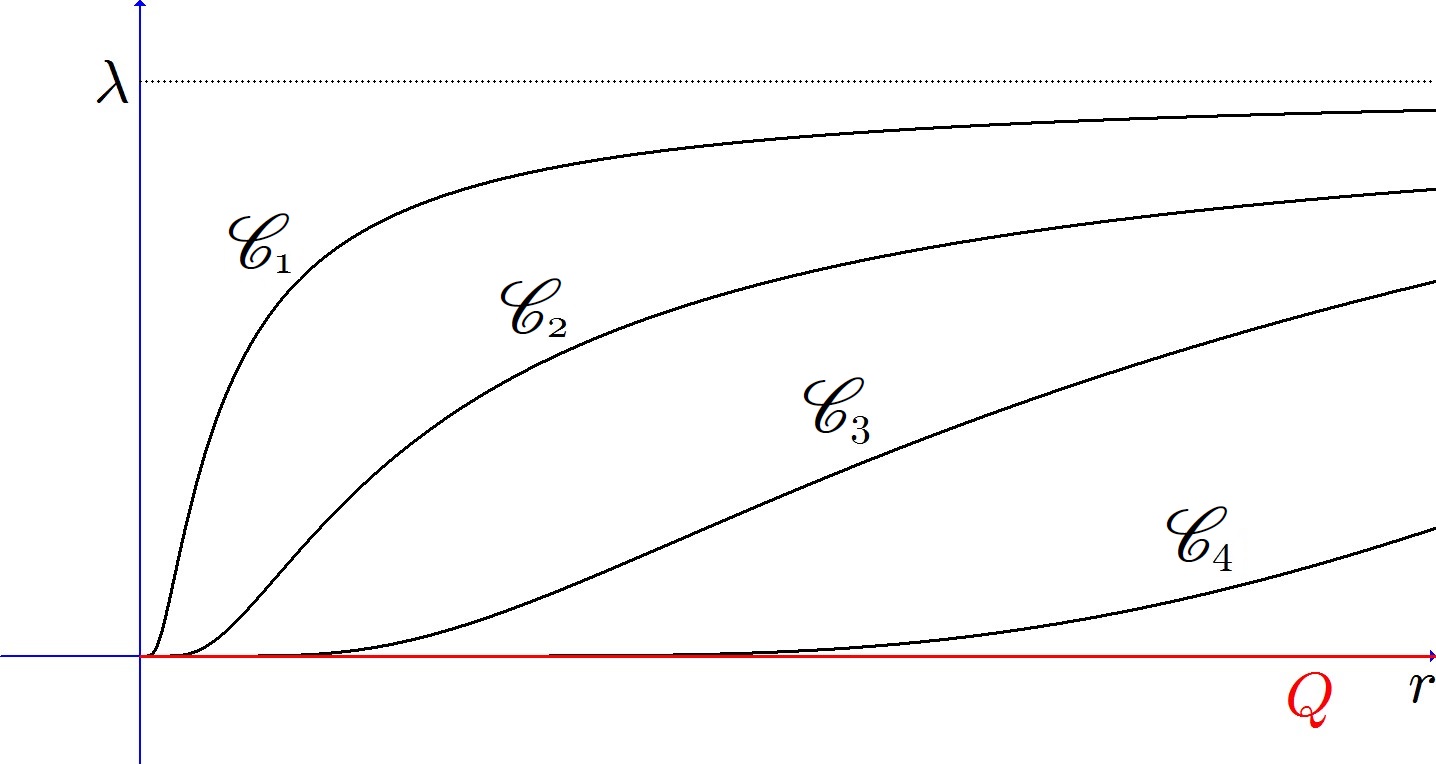}\label{fig:Qa}}
\hfill
\subbottom[$0<\mu<\lambda$]{\includegraphics[width=0.3\textwidth]{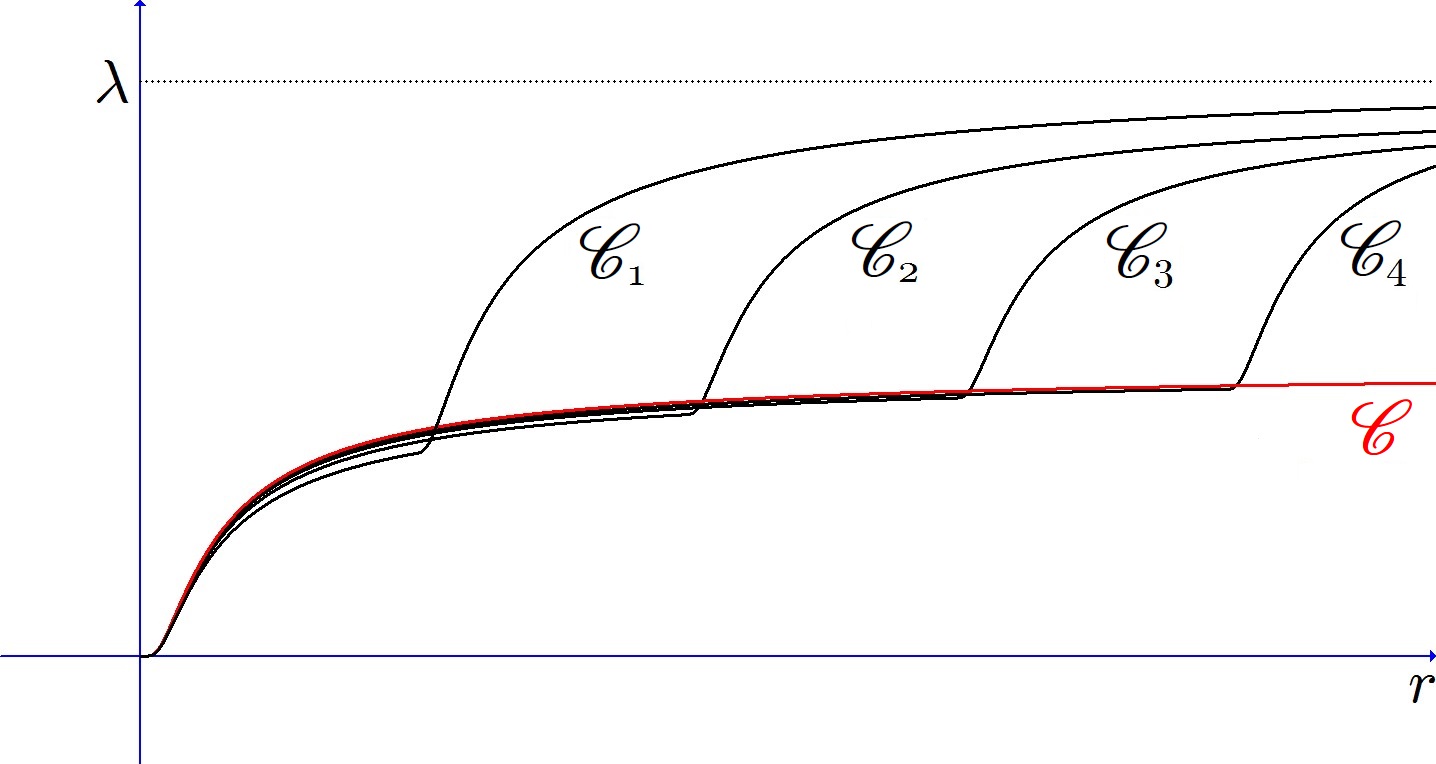}\label{fig:Qb}}
\hfill
\subbottom[$\mu=\lambda$]{\includegraphics[width=0.3\textwidth]{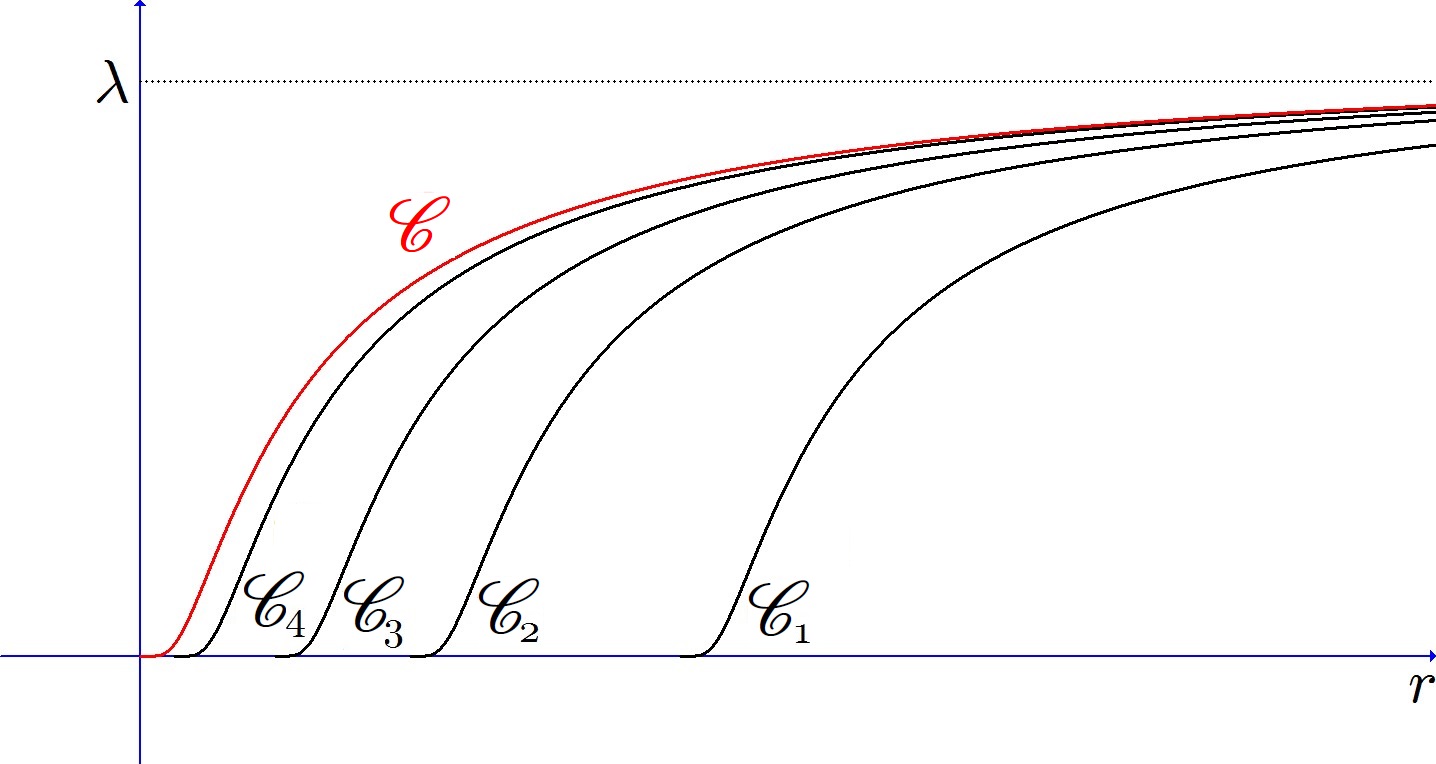}\label{fig:Qc}}
\caption{Examples of possible situations where $\mu=0$, $0<\mu<\lambda$ respectively $\mu=\lambda$. Let us consider the behaviour of $|\psi_n|^2$ in each of these cases as $n$ increases: In the (a)-case the wave function spreads out and diminishes, in the case (b) it splits up into lumps that move further and further away from each other whereas $|\psi_n|^2$ approaches a specific probability distribution in a way that conserves the total probability mass of the particle in the (c)-situation.} \label{fig:Q}
\end{figure}
Of the three scenarios illustrated on Figure \ref{fig:Q} the possibility $\mu=\lambda$ seems to be the most reasonable from a physical point of view and as we will see later the identity $\mu=\lambda$ does indeed hold true. We will basically prove this by ruling out the two other alternatives shown on Figure \ref{fig:Q}. 

We begin by proving that it is impossible for $\mu$ to be equal to $0$. This will be done by first establishing the following lower bound on  $\mathscr{E}_j^{\bm{v}}(\psi_n,\bm{A}_n)$.
\begin{align}
\mathscr{E}_j^{\bm{v}}(\psi_n,\bm{A}_n)\geq -Q^2\frac{\bm{v}^2}{c^2-\bm{v}^2}\int_{\mathbb{R}^3}\int_{\mathbb{R}^3}\frac{|\psi_n(\bm{x})|^2|\psi_n(\bm{y})|^2}{|\bm{x}-\bm{y}|}\,\mathrm{d}\bm{x}\,\mathrm{d}\bm{y}-\frac{m\bm{v}^2}{2}\lambda.\label{lowerboundfieldenergys}
\end{align}
This means that we can control $\mathscr{E}_j^{\bm{v}}(\psi_n,\bm{A}_n)$ by information on the wave functions $\psi_n$ alone -- we remember from Figure \ref{fig:Qa} that the case $\mu=0$ would morally correspond to the eventual disappearance of these wave functions. So in that case we expect the first term on the right hand side of \eqref{lowerboundfieldenergys} to disappear in the large $n$ limit. Our strategy will therefore be to show that the identity $\mu=0$ would violate the inequality in Lemma \ref{negativeI} stating that $I_j^{\lambda}$ is \emph{strictly} less than $-\frac{m\bm{v}^2}{2}$.
\begin{lemma}\label{not0}
Let $j\in\{\mathrm{S},\mathrm{P}\}$, $\bm{v}\in\mathbb{R}^3$ with $0<|\bm{v}|<c$ as well as $\lambda\in \Lambda_j^{\bm{v}}$ be given and consider a minimizing sequence $\bigl((\psi_n,\bm{A}_n)\bigr)_{n\in\mathbb{N}}\subset \mathcal{S}_\lambda$ for $\mathscr{E}_j^{\bm{v}}$. Define for $n\in\mathbb{N}$ the concentration function $\mathscr{C}_n$ by \eqref{Qndef} and let $\mathscr{C}$ be the pointwise limit of \emph{(}a subsequence of\emph{)} $(\mathscr{C}_n)_{n\in\mathbb{N}}$. Then $\mu=\lim_{r\to\infty}\mathscr{C}(r)$ is different from $0$.
\end{lemma}
\begin{proof}
The estimate \eqref{lowerboundfieldenergys} is actually quite rough because in the first step towards obtaining it we simply dispense with the kinetic energy term on the right hand side of \eqref{Eexp2}, resulting in
\begin{align}
\mathscr{E}_j^{\bm{v}}(\psi_n,\bm{A}_n)&\geq \mathscr{G}_n(\bm{A}_n)-\frac{m\bm{v}^2}{2}\lambda,\label{Energyestimatealphaneq0}
\end{align}
where $\mathscr{G}_n:D^1\to\mathbb{R}$ is defined by
\begin{align*}
\mathscr{G}_n(\bm{D})=\frac{1}{8\pi}\Bigl(1-\frac{\bm{v}^2}{c^2}\Bigr)\|\nabla \otimes\bm{D}\|_{L^2}^2-Q\frac{\bm{v}}{c}\cdot\int_{\mathbb{R}^3}\bm{D}(\bm{x})|\psi_n(\bm{x})|^2\,\mathrm{d}\bm{x}.
\end{align*}
This functional is bounded from below since applying the Sobolev and Hölder inequalities as well as optimizing in each of the variables $\bigl\|D^1\bigr\|_{L^6}$, $\bigl\|D^2\bigr\|_{L^6}$ and $\bigl\|D^3\bigr\|_{L^6}$ gives that for any $\bm{D}\in D^1$
\begin{align}
\mathscr{G}_n(\bm{D})&\geq \frac{1}{16\pi}\Bigl(1-\frac{\bm{v}^2}{c^2}\Bigr)\|\nabla \otimes\bm{D}\|_{L^2}^2-4\pi K_S^3Q^2 C\lambda^{\frac{3}{2}}\frac{\bm{v}^2}{c^2-\bm{v}^2}\label{Fbddfrombelow}
\end{align}
so it seems straightforward to meet our intention of obtaining a lower bound on $\mathscr{E}_j^{\bm{v}}(\psi_n,\bm{A}_n)$ only depending on $\psi_n$ -- we can simply estimate the term $\mathscr{G}_n(\bm{A}_n)$ appearing on the right hand side of \eqref{Energyestimatealphaneq0} by $\inf_{\bm{D}\in D^1}\mathscr{G}_n(\bm{D})$. Therefore it will be worthwhile for us to spend some time studying the properties of this infimum.

We first show the existence of a minimizer $\bm{D}_n$ for $\mathscr{G}_n$. This will be done by the direct method in the calculus of variations so consider a minimizing sequence for $\mathscr{G}_n$, i.e. a sequence $\bigl(\bm{D}^k_n\bigr)_{k\in\mathbb{N}}$ of $D^1$-functions such that $\bigl(\mathscr{G}_n\bigl(\bm{D}_n^k\bigr)\bigr)_{k\in\mathbb{N}}$ converges to $\inf_{\bm{D}\in D^1}\mathscr{G}_n(\bm{D})$. Then \eqref{Fbddfrombelow} together with the Sobolev inequality gives that the sequence $\bigl(\bm{D}_n^k\bigr)_{k\in\mathbb{N}}$ is bounded in the reflexive Banach space $L^6$ and in addition that $\bigl(\nabla \otimes\bm{D}_n^k\bigr)_{k\in\mathbb{N}}$ is bounded in the Hilbert space $L^2$. Thereby the Banach-Alaoglu theorem guarantees the existence of a subsequence of $\bigl(\bm{D}_n^k\bigr)_{k\in\mathbb{N}}$ converging weakly in $L^6$ to some $\bm{D}_n$ and by passing to yet another subsequence, $\bigl(\nabla \otimes\bm{D}_n^k\bigr)_{k\in\mathbb{N}}$ converges weakly in $L^2$ to some $\bm{D}_n'$. But then we have $\bm{D}_n^k\xrightarrow[k\to\infty]{} \bm{D}_n$ and $\nabla \otimes\bm{D}_n^k\xrightarrow[k\to\infty]{} \bm{D}_n'$ in the distribution sense, whereby we must have $\nabla \otimes\bm{D}_n=\bm{D}_n'$. In other words, we have (after passing to a subsequence)
\begin{align}
\bm{D}_n^k\xrightharpoonup[k\to\infty]{} \bm{D}_n\textrm{ in }L^6\quad\textrm{and}\quad \nabla\otimes \bm{D}_n^k\xrightharpoonup[k\to\infty]{}\nabla \otimes\bm{D}_n\textrm{ in }L^2.\label{toyconvergences}
\end{align}
That $|\psi_n|^2\in L^{\frac{6}{5}}$ implies together with the first convergence in \eqref{toyconvergences} that $\lim_{k\to\infty}\bm{v}\cdot\int_{\mathbb{R}^3}\bm{D}_n^k(\bm{x})|\psi_n(\bm{x})|^2\,\mathrm{d}\bm{x}$ is equal to $\bm{v}\cdot\int_{\mathbb{R}^3}\bm{D}_n(\bm{x})|\psi_n(\bm{x})|^2\,\mathrm{d}\bm{x}$ and the second convergence in \eqref{toyconvergences} gives together with the weak lower semicontinuity \cite[Theorem 2.11]{LL} of $\|\cdot\|_{L^2}$ that the quantity $\liminf_{k\to\infty}\bigl\|\nabla\otimes \bm{D}_n^k\bigr\|_{L^2}^2$ is at least $\bigl\|\nabla \otimes\bm{D}_n\bigr\|_{L^2}^2$. Thereby
\begin{align*}
\inf_{\bm{D}\in D^1}\mathscr{G}_n(\bm{D})&= \frac{1}{8\pi}\Bigl(1-\frac{\bm{v}^2}{c^2}\Bigr)\liminf_{k\to\infty}\bigl\|\nabla \otimes \bm{D}_n^k\bigr\|_{L^2}^2-Q\frac{\bm{v}}{c}\cdot\int_{\mathbb{R}^3}\!\!\bm{D}_n(\bm{x})|\psi_n(\bm{x})|^2\,\mathrm{d}\bm{x}\\
&\geq \mathscr{G}_n(\bm{D}_n)
\end{align*}
and so we must have $\inf_{\bm{D}\in D^1}\mathscr{G}_n(\bm{D})=\mathscr{G}_n(\bm{D}_n)$. 
Then the functional derivative $\frac{\delta\mathscr{G}_n}{\delta \bm{D}}$ must take the value $0$ in the point $\bm{D}_n$, which implies that $\bm{D}_n$ satisfies the Poisson equation
\begin{align}
-\Delta \bm{D}_n=4\pi Q\frac{c\bm{v}}{c^2-\bm{v}^2}|\psi_n|^2\label{Poissonw0}
\end{align}
in the distribution sense. The function on the right hand side is contained in $L^1\cap L^3$ and has gradient in $L^1\cap L^{\frac{5}{4}}$ so according to Lemma \ref{Poissonlemma} and Remark \ref{uniquePoisson} we must have
\begin{align*}
\bm{D}_n(\bm{x})=Q\frac{c\bm{v}}{c^2-\bm{v}^2}\int_{\mathbb{R}^3} \frac{|\psi_n(\bm{y})|^2}{|\bm{x}-\bm{y}|}\,\mathrm{d}\bm{y}
\end{align*}
for almost every $\bm{x}\in\mathbb{R}^3$. Consequently, we can continue the estimate \eqref{Energyestimatealphaneq0} and get \eqref{lowerboundfieldenergys}.

We now realize that for almost all $\bm{y}\in\mathbb{R}^3$ and all choices of positive numbers $r$ and $R$ satisfying $r<R$ we have
\begin{align*}
\int_{\mathbb{R}^3}\frac{|\psi_n(\bm{x})|^2}{|\bm{x}-\bm{y}|}\,\mathrm{d}\bm{x}
&\leq \|\psi_n\|_{L^6}^2\left\|1_{\mathcal{B}(\bm{0},r)}\frac{1}{|\cdot|}\right\|_{L^{\frac{3}{2}}}+\frac{1}{r}\|1_{\mathcal{B}(\bm{y},R)}\psi_n\|_{L^2}^2+\frac{1}{R}\|\psi_n\|_{L^2}^2\\
&\leq \Bigl(\frac{8\pi}{3}\Bigr)^{\frac{2}{3}}K_S^2C^2r+\frac{1}{r}\mathscr{C}_n(R)+\frac{1}{R}\lambda;
\end{align*}
this is seen by splitting the integral on the left hand side into contributions from $\mathcal{B}(\bm{y},r)$, $\mathcal{B}(\bm{y},R)\setminus\mathcal{B}(\bm{y},r)$ and $\mathbb{R}^3\setminus\mathcal{B}(\bm{y},R)$. Combining this with \eqref{lowerboundfieldenergys} gives
\begin{align*}
\mathscr{E}_j^{\bm{v}}(\psi_n,\bm{A}_n)\geq -Q^2\frac{\bm{v}^2}{c^2-\bm{v}^2}\lambda\left(\Bigl(\frac{8\pi}{3}\Bigr)^{\frac{2}{3}}K_S^2C^2r+\frac{1}{r}\mathscr{C}_n(R)+\frac{1}{R}\lambda\right)-\frac{m\bm{v}^2}{2}\lambda.\label{lowerboundfieldenergy}
\end{align*}
so sending $n$ to infinity results in
\begin{align*}
I_j^\lambda\geq -Q^2\frac{\bm{v}^2}{c^2-\bm{v}^2}\lambda\left(\Bigl(\frac{8\pi}{3}\Bigr)^{\frac{2}{3}}K_S^2C^2r+\frac{1}{r}\mathscr{C}(R)+\frac{1}{R}\lambda\right)-\frac{m\bm{v}^2}{2}\lambda.
\end{align*}
Under the assumption that $\mu=0$ we can therefore let $R\to\infty$ and get
\begin{align*}
I_j^\lambda\geq -Q^2\frac{\bm{v}^2}{c^2-\bm{v}^2}\lambda \Bigl(\frac{8\pi}{3}\Bigr)^{\frac{2}{3}}K_S^2C^2r-\frac{m\bm{v}^2}{2}\lambda,
\end{align*}
which sets us in position to take the limit $r\to 0^+$ and obtain the inequality $I_j^\lambda\geq -\frac{m\bm{v}^2}{2}\lambda$, contradicting Lemma \ref{negativeI}.
\end{proof}
We now turn to proving that $\mu\notin (0,\lambda)$, which will again be done using the method of proof by contradiction. Remember from Figure \ref{fig:Qb} that if $\mu\in (0,\lambda)$ we expect the wave function to split up into lumps that move further and further away from each other as $n$ increases. It seems reasonable that these lumps will eventually be so far apart that the interaction between them is negligible, whereby we can practically consider them as independent systems. Given a term $(\psi_n,\bm{A}_n)$ of the minimizing sequence our strategy will therefore be to construct a pair $(\psi^{\mathrm{i}}_n,\bm{A}^{\mathrm{i}}_n)$ which is `almost' an element of $\mathcal{S}_{\mu}$ and a pair $(\psi^{\mathrm{o}}_n,\bm{A}^{\mathrm{o}}_n)$ `almost' belonging to $\mathcal{S}_{\lambda-\mu}$ such that $\mathscr{E}_j^{\bm{v}}(\psi^{\mathrm{i}}_n,\bm{A}^{\mathrm{i}}_n)+\mathscr{E}_j^{\bm{v}}(\psi^{\mathrm{o}}_n,\bm{A}^{\mathrm{o}}_n)$ is at most $\mathscr{E}_j^{\bm{v}}(\psi_n,\bm{A}_n)$ (up to a small error). A limiting argument will then give a conclusion contradicting \eqref{subadditivity}. The splitting will of course be done by using cut-off functions, so let us first introduce some mappings $\chi^{\mathrm{i}}$ and $\chi^{\mathrm{o}}$ (`$\mathrm{i}$' for `inner' and `$\mathrm{o}$' for `outer') with the following properties: The supports of $\chi^{\mathrm{i}}\in C_0^\infty(\mathbb{R}^3)$ and $\chi^{\mathrm{o}}\in C^\infty(\mathbb{R}^3)$ are disjoint and
\begin{align*}
\chi^{\mathrm{i}}(\bm{x})\begin{cases}=1&\textrm{for }|\bm{x}|\leq 1,\\\in[0,1]&\textrm{for }1<|\bm{x}|<2,\\=0&\textrm{for }|\bm{x}|\geq 2,\end{cases}\quad\textrm{and}\quad \chi^{\mathrm{o}}(\bm{x})\begin{cases}=0&\textrm{for }|\bm{x}|\leq 1,\\\in[0,1]&\textrm{for }1<|\bm{x}|<2,\\=1&\textrm{for }|\bm{x}|\geq 2.\end{cases}
\end{align*}
\begin{figure}[ht]
	\centering
		\includegraphics[width=0.6\textwidth]{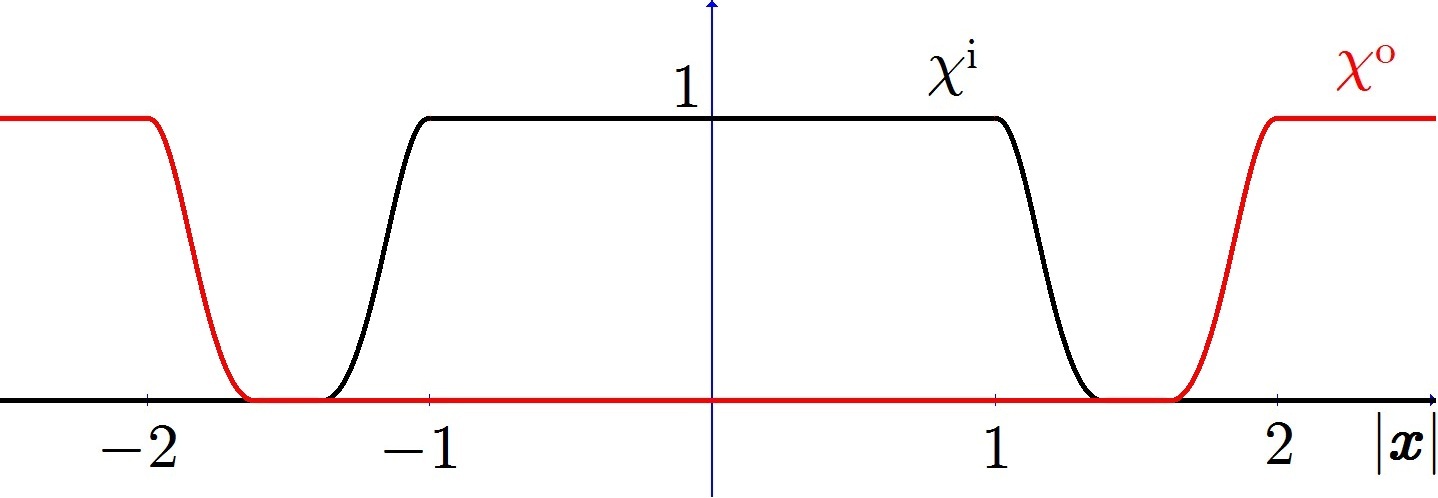}
	\caption{Possible choices for $\chi^{\mathrm{i}}$ and $\textcolor{red}{\chi^{\mathrm{o}}}$.}
	\label{fig:Chis}
\end{figure}
\begin{lemma}\label{munotin0lambda}
Consider $j\in\{\mathrm{S},\mathrm{P}\}$, $\bm{v}\in\mathbb{R}^3$ with $0<|\bm{v}|<c$ and $\lambda\in\Lambda_j^{\bm{v}}$. Let also $\bigl((\psi_n,\bm{A}_n)\bigr)_{n\in\mathbb{N}}\subset\mathcal{S}_{\lambda}$ be a minimizing sequence for $\mathscr{E}_j^{\bm{v}}$, define $\mathscr{C}_n$ by \eqref{Qndef} for $n\in\mathbb{N}$ and consider the pointwise limit $\mathscr{C}$ of \emph{(}a subsequence of\emph{)} $(\mathscr{C}_n)_{n\in\mathbb{N}}$. Then $\mu=\lim_{r\to\infty}\mathscr{C}(r)$ is not contained in $(0,\lambda)$.
\end{lemma}
\begin{proof}
Suppose that $\mu\in (0,\lambda)$. On the basis of $\psi_n$ we want to construct a function $\psi^{\mathrm{i}}_n$ whose $L^2$-norm squared is close to $\mu$, so to which region of space should we localize $\psi_n$? The answer is of course encoded in the concentration function of $\psi_n$, so more precisely: Given
\begin{align}
0<\varepsilon<\min\left\{\frac{2^53^4C^6K_S^6Q^4}{m^3c^4\bm{v}^2},4m\bm{v}^2\mu,4m\bm{v}^2(\lambda-\mu),\frac{2^{\frac{11}{4}}3C^{\frac{3}{2}}K_S^{\frac{3}{2}}|Q\bm{v}|}{c}(\lambda-\mu)^{\frac{3}{4}}\right\}\label{choiceofepsilon}
\end{align}
we choose (by the definition of $\mu$) a number
\begin{align}
R&>\max\left\{2^{\frac{3}{2}}\hbar\sqrt{\frac{\lambda}{\varepsilon m}}\|\nabla\chi^{\ell}\|_{L^\infty},\frac{16C\hbar\lambda^{\frac{1}{2}}}{\varepsilon m}\|\nabla\chi^{\ell}\|_{L^\infty}\,\middle|\, \ell\in\{\mathrm{i},\mathrm{o}\}\right\}\nonumber\\
&\hspace{2cm}\vee \ \frac{2^2 3^4K_S^2 C^2Q^2\lambda^2\bm{v}^2}{\varepsilon^2\sqrt{\pi}c^2}\bigl(\max\bigl\{\bigl\|\nabla \chi^{\mathrm{i}}\bigr\|_{L^{12}},\bigl\|\nabla\chi^{\mathrm{o}}\bigr\|_{L^3}\bigr\}\bigr)^2\label{Rest}
\end{align}
such that $\mu-\frac{\varepsilon^{4/3}c^{4/3}}{2^{11/3}3^{4/3}C^2K_S^2|Q\bm{v}|^{4/3}}<\mathscr{C}(R)$. As a first step we will consider $n$'s so large that
\begin{align}
\mu-\frac{\varepsilon^{\frac{4}{3}}c^{\frac{4}{3}}}{2^{\frac{11}{3}}3^{\frac{4}{3}}C^2K_S^2|Q\bm{v}|^{\frac{4}{3}}}<\mathscr{C}_n(R)<\mu+\frac{\varepsilon^{\frac{4}{3}}c^{\frac{4}{3}}}{2^{\frac{11}{3}}3^{\frac{4}{3}}C^2K_S^2|Q\bm{v}|^{\frac{4}{3}}}.\label{Qnlower}
\end{align}
Here, the upper bound on $\mathscr{C}_n(R)$ is strictly speaking redundant, since we will later obtain a better upper bound by considering even larger values of $n$ -- but already at this point it is advantageous to think of $\mathscr{C}_n(R)$ as being close to $\mu$. We should not just perceive $\mathscr{C}_n(R)$ as being an abstract supremum -- it is in fact the probability mass of the particle in the vicinity of some point in space. Because $\psi_n\in L^2$ the continuous function $\bm{y}\mapsto \int_{\mathcal{B}(\bm{y},R)}|\psi_n(\bm{x})|^2\,\mathrm{d}\bm{x}$ will namely approach zero as $|\bm{y}|\to \infty$, whereby we can choose a point $\bm{y}_n\in\mathbb{R}^3$ such that
\begin{align}
\mathscr{C}_n(R)=\int_{\mathcal{B}(\bm{y}_n,R)}|\psi_n(\bm{x})|^2\,\mathrm{d}\bm{x}.\label{Qn=}
\end{align}
So in the ball $\mathcal{B}(\bm{y}_n,R)$ we have found a $\psi_n$-lump whose probability mass is essentially $\mu$. The other lumps are expected to move away as $n$ increases, so for large $n$ there should be a large area around $\mathcal{B}(\bm{y}_n,R)$ where $\psi_n$ has essentially no probability mass. As a consequence we can construct the function $\psi_n^{\mathrm{o}}$ by cutting away the values of $\psi_n$ on a ball centered at $\bm{y}_n$ with quite a large radius. It turns out that we can in fact choose this radius on the form $2^{k_n}R$, where the sequence $(k_n)_{n\in\mathbb{N}}$ of integers satisfies
\begin{enumerate}
\item[(I)] $k_n\to \infty$ for $n\to\infty$,
\item[(II)] $\mathscr{C}_n(2^{k_n}R)\leq \mu+\frac{\varepsilon^{4/3}c^{4/3}}{2^{11/3}3^{4/3}C^2K_S^2|Q\bm{v}|^{4/3}}$ for all $n\in\mathbb{N}$.
\end{enumerate}
One can namely easily verify that the sequence of numbers
\begin{align*}
k_n=\left\lfloor \log_2\frac{\sup \mathscr{C}_n^{-1}\left(\left(0,\mu+\frac{\varepsilon^{4/3}c^{4/3}}{2^{11/3}3^{4/3}C^2K_S^2|Q\bm{v}|^{4/3}}\right]\right)}{2R}\right\rfloor
\end{align*}
has the desired properties, where $\lfloor \cdot\rfloor$ denotes the floor function and $\log_2$ denotes the binary logarithm\footnote{The floor function is $x\mapsto\max\{m\in\mathbb{Z}\mid m\leq x\}$ and the binary logarithm is $x\mapsto \frac{\log(x)}{\log(2)}$, where $\log$ denotes the natural logarithm.}. Thus, we will construct $\psi_n^{\mathrm{i}}$ and $\psi_n^{\mathrm{o}}$ by multiplication with the cut-off functions given by
\begin{align*}
\chi_n^{\mathrm{i},\psi}(\bm{x})=\chi^{\mathrm{i}}\Bigl(\frac{\bm{x}-\bm{y}_n}{R}\Bigr)\quad\textrm{respectively}\quad \chi_n^{\mathrm{o},\psi}(\bm{x})=\chi^{\mathrm{o}}\Bigl(\frac{\bm{x}-\bm{y}_n}{2^{k_n-1}R}\Bigr)
\end{align*}
for $\bm{x}\in\mathbb{R}^3$. Let us emphasize that we use the superscript $\psi$ because these functions will be used to cut the wave function $\psi$ into the two pieces $\psi_n^{\mathrm{i}}$ and $\psi_n^{\mathrm{o}}$ -- later we will define corresponding cut-off functions $\chi_n^{\mathrm{i},\bm{A}}$ and $\chi_n^{\mathrm{o},\bm{A}}$ to cut $\bm{A}$ into two pieces $\bm{A}_n^{\mathrm{i}}$ and $\bm{A}_n^{\mathrm{o}}$.

Let us now do this splitting of the $\bm{A}_n$-field into $\bm{A}_n^{\mathrm{i}}$- and $\bm{A}_n^{\mathrm{o}}$-fields. We will aim to make the cuts in the big gap between $\mathcal{B}(\bm{y}_n,2R)$ and $\mathbb{R}^3\setminus\mathcal{B}(\bm{y}_n,2^{k_n-1}R)$, where the functions $\psi_n^{\mathrm{i}}$ and $\psi_n^{\mathrm{o}}$ are guaranteed to vanish. So we decompose space into the disjoint union $\mathbb{R}^3=\mathcal{B}(\bm{y}_n,R)\cup\bigl(\bigcup_{m=1}^\infty \mathcal{A}_n^m\bigr)$, where
\begin{align*}
\mathcal{A}_n^m=\bigl\{\bm{x}\in\mathbb{R}^3\,\big|\, 2^{m-1}R\leq |\bm{x}-\bm{y}_n|< 2^{m} R\bigr\}
\end{align*}
for $m\in\mathbb{N}$ (see Figure \ref{fig:Shells}).
\begin{figure}[ht]
	\centering
		\includegraphics[width=0.6\textwidth]{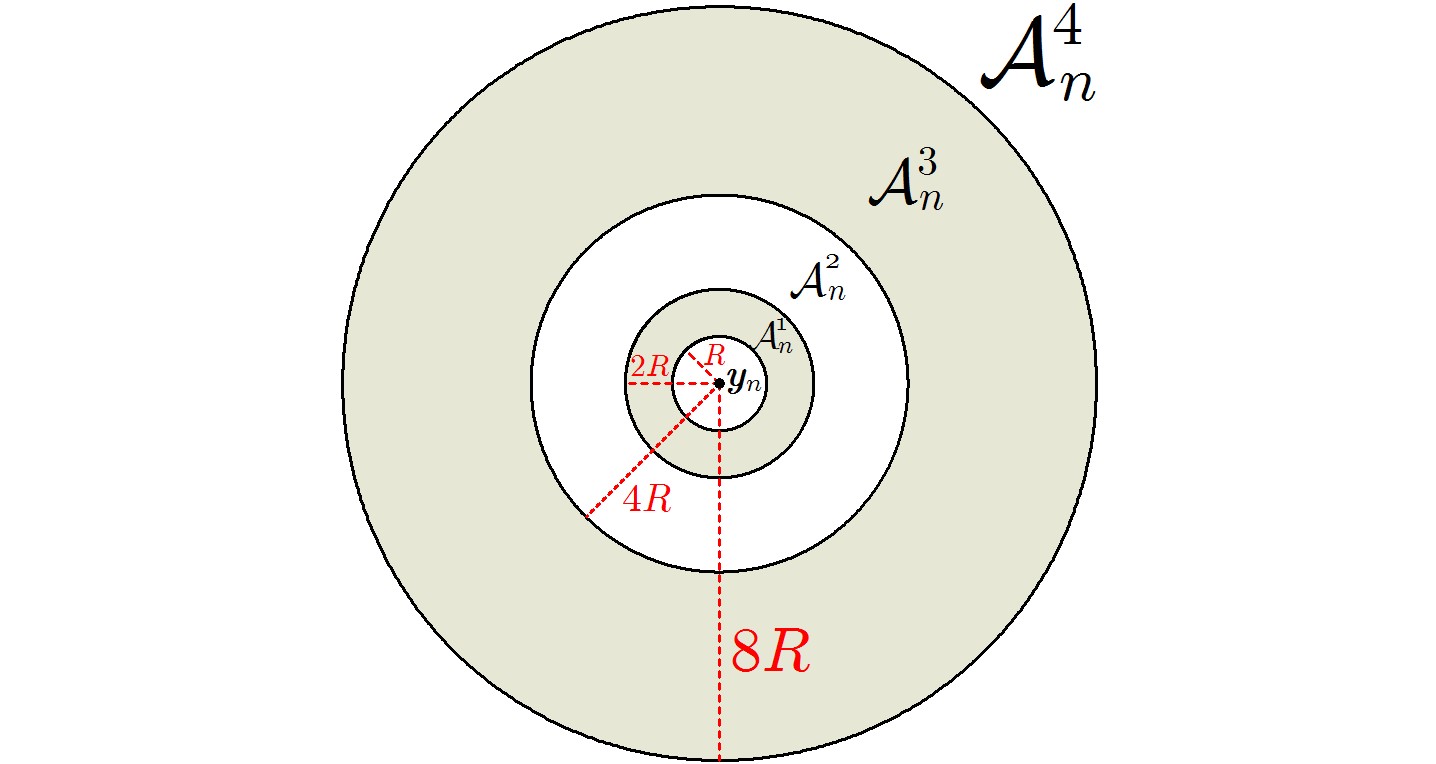}
	\caption{The two-dimensional analogues of the spherical shells $\mathcal{A}_n^0,\mathcal{A}_n^1,\ldots$.}
	\label{fig:Shells}
\end{figure}
By point (I) from above we have $k_n\geq 4$ for $n$ sufficiently large and for such $n$'s there must exist a number $m_n$ in the set $\{2,\ldots,k_n-2\}$ such that $\bigl\|1_{\mathcal{A}_n^{m_n}}\bm{A}_n\bigr\|_{L^6}^6\leq (k_n-3)^{-1}\bigl\|1_{\mathcal{A}_n^{2}\cup\cdots\cup\mathcal{A}_n^{k_n-2}}\bm{A}_n\bigr\|_{L^6}^6$ holds true, whereby we have for $n$ sufficiently large that
\begin{align}
&\bigl\|1_{\mathcal{A}_n^{m_n}}\bm{A}_n\bigr\|_{L^6}\nonumber\\
&<\min\Bigl\{\frac{1}{4\|\nabla\chi^{\ell}\|_{L^3}}\Bigl(\sqrt{C^2+\frac{\varepsilon\pi c^2}{c^2+\bm{v}^2}}-C\Bigr),\frac{1}{2\|\nabla\chi^{\ell}\|_{L^3}}\sqrt{\frac{\varepsilon\pi c^2}{c^2+\bm{v}^2}}\Big|\ell\in\{\mathrm{i},\mathrm{o}\}\Bigr\}.\label{controlA}
\end{align}
In this way we can control $\bm{A}_n$ on $\mathcal{A}_n^{m_n}$, so we will define $\bm{A}_n^{\mathrm{i}}$ and $\bm{A}_n^{\mathrm{o}}$ using the cut-off functions
\begin{align*}
\chi_{n}^{\mathrm{i},\bm{A}}(\bm{x})=\chi^{\mathrm{i}}\Bigl(\frac{\bm{x}-\bm{y}_n}{2^{{m_n-1}}R}\Bigr)\quad\textrm{and}\quad \chi_{n}^{\mathrm{o},\bm{A}}(\bm{x})=\chi^{\mathrm{o}}\Bigl(\frac{\bm{x}-\bm{y}_n}{2^{{m_n-1}}R}\Bigr).
\end{align*}
\begin{figure}[ht]
	\centering
		\includegraphics[width=0.9\textwidth]{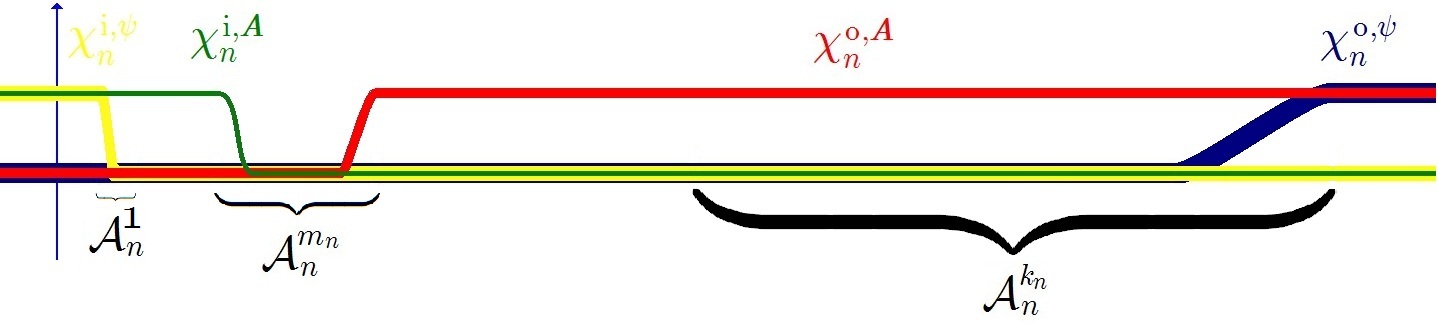}
	\caption{The distance to $\bm{y}_n$ is measured along the first axis.}
	\label{fig:Cut-off}
\end{figure}
More precisely, we will for $\ell\in\{\mathrm{i},\mathrm{o}\}$ introduce the mapping $u_n^\ell:\mathbb{R}^3\to\mathbb{R}$ given by
\begin{align*}
u_n^\ell(\bm{x})=\frac{1}{4\pi}\int_{\mathbb{R}^3}\frac{1}{|\bm{x}-\bm{y}|}\mathrm{div}\bigl(\chi_n^{\ell,\bm{A}}\bm{A}_n\bigr)(\bm{y})\,\mathrm{d}\bm{y}\textrm{ for almost every }\bm{x}\in\mathbb{R}^3
\end{align*}
and define $\psi_n^{\mathrm{i}}$, $\psi_n^{\mathrm{o}}$, $\bm{A}_n^{\mathrm{i}}$ and $\bm{A}_n^{\mathrm{o}}$ by
\begin{align*}
\psi_n^\ell = \mathrm{e}^{\frac{iQ}{\hbar c} u_n^{\ell}}\chi_n^{\ell,\psi}\psi_n\quad\textrm{and}\quad \bm{A}_n^\ell=\chi_n^{\ell,\bm{A}}\bm{A}_n+\nabla u_n^{\ell}.
\end{align*}
We observe that $\mathrm{div}\bigl(\chi_n^{\ell,\bm{A}}\bm{A}_n\bigr)=\nabla\chi_n^{\ell,\bm{A}}\cdot\bm{A}_n$ is contained in $H^1$ and has compact support so from Lemma \ref{Poissonlemma} we obtain that $\psi_n^\ell\in H^1$, $\bm{A}_n^\ell\in D^1$ with $\mathrm{div}\bm{A}_n^\ell=0$ and
\begin{align}
\nabla u_n^\ell(\bm{x})=-\frac{1}{4\pi}\int_{\mathbb{R}^3}\frac{\bm{x}-\bm{y}}{|\bm{x}-\bm{y}|^3}\mathrm{div}\bigl(\chi_n^{\ell,\bm{A}}\bm{A}_n\bigr)(\bm{y})\,\mathrm{d}\bm{y}\textrm{ for almost every }\bm{x}\in\mathbb{R}^3.\label{nablau}
\end{align}
Moreover, $\psi_n^{\mathrm{i}}$ and $\psi_n^{\mathrm{o}}$ satisfy
\begin{align}
\max\left\{\left|\mu-\int_{\mathbb{R}^3}|\psi_n^{\mathrm{i}}(\bm{x})|^2\,\mathrm{d}\bm{x}\right|,\left|\lambda-\mu-\int_{\mathbb{R}^3}|\psi_n^{\mathrm{o}}(\bm{x})|^2\,\mathrm{d}\bm{x}\right|\right\}<\frac{\varepsilon^{\frac{4}{3}}c^{\frac{4}{3}}}{2^{\frac{11}{3}}3^{\frac{4}{3}}C^2K_S^2|Q\bm{v}|^{\frac{4}{3}}}.\label{psiest}
\end{align}
which follows from \eqref{Qnlower}, (II) as well as the estimates
\begin{align*}
\mathscr{C}_n(R)\leq \int_{\mathbb{R}^3}|\psi_n^{\mathrm{i}}(\bm{x})|^2\,\mathrm{d}\bm{x}\leq\int_{\mathcal{B}(\bm{y}_n,2R)}|\psi_n(\bm{x})|^2\,\mathrm{d}\bm{x}\leq \mathscr{C}_n(2R)\leq \mathscr{C}_n(2^{k_n}R)
\end{align*}
and
\begin{align*}
\mathscr{C}_n(R)\leq \lambda-\int_{\mathbb{R}^3}|\psi_n^{\mathrm{o}}(\bm{x})|^2\,\mathrm{d}\bm{x}\leq\int_{\mathcal{B}(\bm{y}_n,2^{k_n}R)}|\psi_n(\bm{x})|^2\,\mathrm{d}\bm{x}\leq \mathscr{C}_n(2^{k_n}R).
\end{align*}
In the motivational remarks made above Lemma \ref{munotin0lambda} we mentioned the desire to construct $\psi_n^{\ell}$ and $\bm{A}_n^{\ell}$ in such a way that they `almost' satisfy $\bigl(\psi_n^{\mathrm{i}},\bm{A}_n^{\mathrm{i}}\bigr)\in \mathcal{S}_{\mu}$ and $\bigl(\psi_n^{\mathrm{o}},\bm{A}_n^{\mathrm{o}}\bigr)\in \mathcal{S}_{\lambda-\mu}$. The precise meaning of this informal statement is that the pairs $\bigl(\psi_n^{\ell},\bm{A}_n^{\ell}\bigr)\in H^1\times D^1$ have the properties $\mathrm{div}\bm{A}_n^{\ell}=0$ and \eqref{psiest}. 

The next step in our argument is to show that
\begin{align}
\mathscr{E}_j^{\bm{v}}(\psi_n,\bm{A}_n)\geq \mathscr{E}_j^{\bm{v}}(\psi_n^{\mathrm{i}},\bm{A}_n^{\mathrm{i}})+\mathscr{E}_j^{\bm{v}}(\psi_n^{\mathrm{o}},\bm{A}_n^{\mathrm{o}})-\varepsilon.\label{Energysubadditive}
\end{align}
We begin by estimating the $\frac{1}{2m}\|\nabla_{j,\bm{A}_n+\frac{mc}{Q}\bm{v}}\psi_n\|_{L^2}^2$-term on the right hand side of \eqref{Eexp2}. For this we observe that $\chi_n^{\ell,\bm{A}}\chi_n^{\ell,\psi}=\chi_n^{\ell,\psi}$, whereby we can rewrite
\begin{align*}
\mathrm{e}^{-\frac{iQ}{\hbar c} u_n^\ell}\nabla_{j,\bm{A}_n^{\ell}+\frac{mc}{Q}\bm{v}}\psi_n^{\ell}=\nabla_{j,\bm{0}}\chi_n^{\ell,\psi}\psi_n+\chi_n^{\ell,\psi}\nabla_{j,\bm{A}_n+\frac{mc}{Q}\bm{v}}\psi_n,
\end{align*}
which allows us to apply \eqref{sigmavector}, \eqref{Rest} and Remark \ref{bdn} to obtain
\begin{align*}
&\|\nabla_{j,\bm{A}_n^{\ell}+\frac{mc}{Q}\bm{v}}\psi_n^{\ell}\|_{L^2}^2\\
&\leq \frac{\hbar^2\lambda}{R^2}\|\nabla\chi^{\ell}\|_{L^\infty}^2+\|\chi_n^{\ell,\psi}\nabla_{j,\bm{A}_n+\frac{mc}{Q}\bm{v}}\psi_n\|_{L^2}^2+\frac{2C\hbar \lambda^{\frac{1}{2}}}{R}\|\nabla\chi^{\ell}\|_{L^\infty}\\
&\leq \|\chi_n^{\ell,\psi}\nabla_{j,\bm{A}_n+\frac{mc}{Q}\bm{v}}\psi_n\|_{L^2}^2+\frac{\varepsilon m}{4}.
\end{align*}
and consequently
\begin{align}
\frac{1}{2m}\|\nabla_{j,\bm{A}_n+\frac{mc}{Q}\bm{v}}\psi_n\|_{L^2}^2&\geq\frac{1}{2m}\sum_{\ell\in\{\mathrm{i},\mathrm{o}\}}\|\nabla_{j,\bm{A}_n^{\ell}+\frac{mc}{Q}\bm{v}}\psi_n^{\ell}\|_{L^2}^2-\frac{\varepsilon}{4}.\label{split1}
\end{align}
To treat the term $-\frac{Q}{c}(\psi_n,\bm{v}\cdot\bm{A}_n\psi_n)_{L^2}$ appearing on the right hand side of \eqref{Eexp2} we establish two auxiliary estimates: The first estimate
\begin{align*}
&\Bigl|\int_{\mathbb{R}^3}\bm{v}\cdot\bm{A}_n\bigl(|\psi_n|^2-|\psi_n^{\mathrm{i}}|^2-|\psi_n^{\mathrm{o}}|^2\bigr)\,\mathrm{d}\bm{x}\Bigr|\\
&\leq |\bm{v}|\|\bm{A}_n\|_{L^6}\left\|\sqrt{1-(\chi_n^{\mathrm{i},\psi})^2-(\chi_n^{\mathrm{o},\psi})^2}\psi_n\right\|_{L^6}^{\frac{1}{2}}\Bigl(\int_{\mathbb{R}^3}\bigl(|\psi_n|^2-|\psi_n^{\mathrm{i}}|^2-|\psi_n^{\mathrm{o}}|^2\bigr)\,\mathrm{d}\bm{x}\Bigr)^{\frac{3}{4}}\\
&\leq \frac{\varepsilon c}{12|Q|}
\end{align*}
follows from \eqref{psiest}, Hölder's and Sobolev's inequalities. By choosing $n$ large enough we previously made sure that $k_n\geq 4$ and $2\leq m_n\leq k_n-2$ whereby \eqref{nablau}, the Hölder inequality and \eqref{Rest} yield the second auxiliary estimate
\begin{align*}
&\Bigl|\int_{\mathbb{R}^3}\bm{v}\cdot\nabla u_n^{\ell}|\psi_n^{\ell}|^2\,\mathrm{d}\bm{x}\Bigr|\\
&\leq \frac{|\bm{v}|}{4\pi}\int_{\mathbb{R}^3}\int_{\mathcal{A}_n^{m_n}}\frac{1}{|\bm{x}-\bm{y}|^2}|\nabla\chi_n^{\ell,\bm{A}}(\bm{y})||\bm{A}_n(\bm{y})|\,\mathrm{d}\bm{y}\,|\psi_n^\ell(\bm{x})|^2\,\mathrm{d}\bm{x}\\
&\leq
\begin{dcases}
\frac{|\bm{v}|}{4\pi}\left\|1_{\mathcal{B}(\bm{0},(2+2^{m_n})R)}\frac{1}{|\cdot|}\right\|_{L^{\frac{8}{3}}}^2\|\nabla\chi_n^{\mathrm{i},\bm{A}}\|_{L^{12}}\|\bm{A}_n\|_{L^6}\|\psi_n^{\mathrm{i}}\|_{L^2}^2&\textrm{for }\ell=\mathrm{i}\\
\frac{|\bm{v}|}{4\pi}\left\|1_{\mathbb{R}^3\setminus\mathcal{B}(\bm{0},(2^{k_n-1}-2^{m_n})R)}\frac{1}{|\cdot|}\right\|_{L^4}^2\|\nabla\chi_n^{\mathrm{o},\bm{A}}\|_{L^3}\|\bm{A}_n\|_{L^6}\|\psi_n^{\mathrm{o}}\|_{L^2}^2&\textrm{for }\ell=\mathrm{o}
\end{dcases}\\
&\leq \frac{3|\bm{v}|K_S C\lambda}{2\pi^{\frac{1}{4}} R^{\frac{1}{2}}}\max\bigl\{\|\nabla\chi^{\mathrm{i}}\|_{L^{12}},\|\nabla\chi^{\mathrm{o}}\|_{L^3}\bigr\}\\
&<\frac{\varepsilon c}{12|Q|}.
\end{align*}
By combining the two previous estimates with the identity $\chi_n^{\ell,\bm{A}}\chi_n^{\ell,\psi}=\chi_n^{\ell,\psi}$ we obtain
\begin{align}
&|(\psi_n,\bm{v}\cdot\bm{A}_n\psi_n)_{L^2}-(\psi_n^{\mathrm{i}},\bm{v}\cdot\bm{A}_n^{\mathrm{i}}\psi_n^{\mathrm{i}})_{L^2}-(\psi_n^{\mathrm{o}},\bm{v}\cdot\bm{A}_n^{\mathrm{o}}\psi_n^{\mathrm{o}})_{L^2}|\nonumber\\
&=\Bigl|\int_{\mathbb{R}^3}\bm{v}\cdot\bm{A}_n\bigl(|\psi_n|^2-|\psi_n^{\mathrm{i}}|^2-|\psi_n^{\mathrm{o}}|^2\bigr)\,\mathrm{d}\bm{x}-\sum_{\ell\in\{\mathrm{i},\mathrm{o}\}}\int_{\mathbb{R}^3}\bm{v}\cdot\nabla u_n^{\ell}|\psi_n^{\ell}|^2\,\mathrm{d}\bm{x}\Bigr|\nonumber\\
&<\frac{\varepsilon c}{4|Q|}.\label{split2}
\end{align}
Finally, we estimate the $\frac{1}{8\pi}\bigl(\|\nabla\otimes\bm{A}_n\|_{L^2}^2-\bigl\|\bigl(\frac{\bm{v}}{c}\cdot\nabla\bigr)\bm{A}_n\bigr\|_{L^2}^2\bigr)$-term on the right hand side of \eqref{Eexp2} by noting that
\begin{align}
&\|\nabla\otimes\bm{A}_n\|_{L^2}^2-\Bigl\|\Bigl(\frac{\bm{v}}{c}\cdot\nabla\Bigr)\bm{A}_n\Bigr\|_{L^2}^2\nonumber\\
&\geq \sum_{\ell\in\{\mathrm{i},\mathrm{o}\}}\Bigl(\|\chi_n^{\ell,\bm{A}}\nabla\otimes\bm{A}_n\|_{L^2}^2-\Bigl\|\chi_n^{\ell,\bm{A}}\Bigl(\frac{\bm{v}}{c}\cdot\nabla\Bigr)\bm{A}_n\Bigr\|_{L^2}^2\Bigr)\nonumber\\
&\geq\sum_{\ell\in\{\mathrm{i},\mathrm{o}\}}\Bigl(\|\nabla\otimes \bm{A}_n^\ell\|_{L^2}^2-\Bigl\|\Bigl(\frac{\bm{v}}{c}\cdot\nabla\Bigr)\bm{A}_n^\ell\Bigr\|_{L^2}^2\nonumber\\
&\hspace{1.5cm}-2\Bigl(1+\frac{\bm{v}^2}{c^2}\Bigr)\|\nabla\otimes\bm{A}_n^\ell\|_{L^2}\bigl(\|\nabla\chi_n^{\ell,\bm{A}}\otimes \bm{A}_n\|_{L^2}+\|\nabla\otimes \nabla u_n^\ell\|_{L^2}\bigr)\nonumber\\
&\hspace{4.25cm}-2\Bigl(1+\frac{\bm{v}^2}{c^2}\Bigr)\|\nabla\chi_n^{\ell,\bm{A}}\otimes \bm{A}_n\|_{L^2}\|\nabla\otimes\nabla u_n^{\ell}\|_{L^2}\Bigr),\label{Atermenergy}
\end{align}
where we at the second step use the identities
\begin{align}
\chi_n^{\ell,\bm{A}}\nabla\otimes\bm{A}_n&=\nabla\otimes\bm{A}_n^\ell-(\nabla\chi_n^{\ell,\bm{A}}\otimes\bm{A}_n+\nabla\otimes\nabla u_n^\ell),\label{firstnablaidentity}\\
\chi_n^{\ell,\bm{A}}\Bigl(\frac{\bm{v}}{c}\cdot\nabla\Bigr)\bm{A}_n&=\Bigl(\frac{\bm{v}}{c}\cdot\nabla\Bigr)\bm{A}_n^{\ell}-\Bigl(\Bigl(\frac{\bm{v}}{c}\cdot\nabla\Bigr)\chi_n^{\ell,\bm{A}}\bm{A}_n+\Bigl(\frac{\bm{v}}{c}\cdot\nabla\Bigr)\nabla u_n^\ell\Bigr)\nonumber
\end{align}
and the nonnegativity of $\bigl(1-\frac{\bm{v}^2}{c^2}\bigr)(\|\nabla\chi_n^{\ell,\bm{A}}\otimes\bm{A}_n\|^2+\|\nabla\otimes\nabla u_n^\ell\|_{L^2}^2)$. As can be seen by approximating $\nabla u_n^\ell$ in $D^1$ by $C_0^\infty$-functions, applying the Plancherel theorem and using the general vector identity $|\bm{F}|^2|\bm{G}|^2=|\bm{F}\cdot\bm{G}|^2+|\bm{F}\times\bm{G}|^2$ we have $\|\nabla \otimes \nabla u_n^\ell\|_{L^2}^2= \|\mathrm{div}\nabla u_n^\ell\|_{L^2}^2+\|\nabla\times \nabla u_n^\ell\|_{L^2}^2=\|\nabla\chi_n^{\ell,\bm{A}}\cdot\bm{A}_n\|_{L^2}^2$. Consequently, $\|\nabla \otimes \nabla u_n^\ell\|_{L^2}$ and $\|\nabla\chi_n^{\ell,\bm{A}}\cdot\bm{A}_n\|_{L^2}$ have the common upper bound $\|\nabla\chi^{\ell}\|_{L^3}\bigl\|1_{\mathcal{A}_n^{m_n}}\bm{A}_n\bigr\|_{L^6}$ that is small by \eqref{controlA}. Moreover, $\|\nabla\otimes \bm{A}_n^{\ell}\|_{L^2}$ is according to \eqref{firstnablaidentity} bounded from above by $C+2\|\nabla\chi^{\ell}\|_{L^3}\bigl\|1_{\mathcal{A}_n^{m_n}}\bm{A}_n\bigr\|_{L^6}$ and so we can use \eqref{controlA} to continue \eqref{Atermenergy} and get
\begin{align}
&\frac{1}{8\pi}\Bigl(\|\nabla\otimes\bm{A}_n\|_{L^2}^2-\Bigl\|\Bigl(\frac{\bm{v}}{c}\cdot\nabla\Bigr)\bm{A}_n\Bigr\|_{L^2}^2\Bigr)\nonumber\\
&\geq \sum_{\ell\in\{\mathrm{i},\mathrm{o}\}}\frac{1}{8\pi}\Bigl(\|\nabla\otimes\bm{A}_n^\ell\|_{L^2}^2-\Bigl\|\Bigl(\frac{\bm{v}}{c}\cdot\nabla\Bigr)\bm{A}_n^\ell\Bigr\|_{L^2}^2\Bigr)-\frac{\varepsilon}{4}.\label{split3}
\end{align}
Now, \eqref{Energysubadditive} is an immediate consequence of \eqref{split1}, \eqref{split2}, \eqref{split3} and the inequality
\begin{align}
\max\left\{\left|\mu-\int_{\mathbb{R}^3}|\psi_n^{\mathrm{i}}(\bm{x})|^2\,\mathrm{d}\bm{x}\right|,\left|\lambda-\mu-\int_{\mathbb{R}^3}|\psi_n^{\mathrm{o}}(\bm{x})|^2\,\mathrm{d}\bm{x}\right|\right\}<\frac{\varepsilon}{4m\bm{v}^2}\label{psiest2}
\end{align}
that follows from \eqref{choiceofepsilon} and \eqref{psiest}.

Finally, Remark \ref{Idecreasing} and \eqref{psiest2} give
\begin{align*}
\mathscr{E}_j^{\bm{v}}(\psi_n,\bm{A}_n)\geq I_j^{\|\psi_n^{\mathrm{i}}\|_{L^2}^2}+I_j^{\|\psi_n^{\mathrm{o}}\|_{L^2}^2}-\varepsilon\geq I_j^{\mu+\frac{\varepsilon}{4m\bm{v}^2}}+I_j^{\lambda-\mu+\frac{\varepsilon}{4m\bm{v}^2}}-\varepsilon,
\end{align*}
so letting $n$ diverge to infinity produces the estimate
\begin{align*}
I_j^\lambda\geq I_j^{\mu+\frac{\varepsilon}{4m\bm{v}^2}}+I_j^{\lambda-\mu+\frac{\varepsilon}{4m\bm{v}^2}}-\varepsilon.
\end{align*}
By Lemma \ref{Icontinuous} we can therefore take the limit $\varepsilon\to 0^+$ and obtain the inequality $I_j^\lambda\geq I_j^{\mu}+I_j^{\lambda-\mu}$ contradicting \eqref{subadditivity}.
\end{proof}
Combining the Lemmas \ref{not0} and \ref{munotin0lambda} allows us to reach the conclusion that $\lim_{r\to\infty}\mathscr{C}(r)$ is equal to $\lambda$. This is exactly what we need to break the translation invariance of our problem.
\begin{proposition}\label{tightnesssaetning}
Given $j\in\{\mathrm{S},\mathrm{P}\}$, $\bm{v}\in\mathbb{R}^3$ with $0<|\bm{v}|<c$ and $\lambda\in\Lambda_j^{\bm{v}}$ consider a minimizing sequence $\bigl((\psi_n,\bm{A}_n)\bigr)_{n\in\mathbb{N}}\subset \mathcal{S}_{\lambda}$ for $\mathscr{E}_j^{\bm{v}}$. Then there exists a sequence $(\bm{y}_n)_{n\in\mathbb{N}}$ of points in $\mathbb{R}^3$ with the following property: For every $\varepsilon>0$ there exists an $R>0$ such that for all $n\in\mathbb{N}$
\begin{align*}
\bigl\|1_{\mathcal{B}(\bm{y}_n,R)}\psi_n\bigr\|_{L^2}^2\geq \lambda-\varepsilon.
\end{align*}
This property is sometimes expressed by saying that the maps $\bm{x}\mapsto |\psi_n(\bm{x}+\bm{y}_n)|^2$ are \emph{tight}.
\end{proposition}
\begin{proof}
Given $\nu>0$ we will first argue that it is possible to find an $r^\nu>0$ such that
\begin{align}
\mathscr{C}_n(r^\nu)>\lambda-\nu\quad\textrm{for all }n\in\mathbb{N}.\label{greaterthannu}
\end{align}
By the identity $\lim_{r\to\infty}\mathscr{C}(r)=\lambda$ we can namely consider a $\rho>0$ such that $\mathscr{C}(\rho)>\lambda-\nu$ and then we can choose an $N\in\mathbb{N}$ such that $\mathscr{C}_n(\rho)>\lambda-\nu$ for $n>N$, simply because $\lim_{n\to\infty}\mathscr{C}_n(\rho)=\mathscr{C}(\rho)$. Finally, we can use the fact that $\lim_{r\to\infty}\mathscr{C}_n(r)=\lambda$ for each of the finitely many $n$'s in the set $\{1,\ldots,N\}$ to find a $\rho_n>0$ satisfying $\mathscr{C}_n(\rho_n)>\lambda-\nu$. Then because $\mathscr{C}_n$ is a nondecreasing function the inequality \eqref{greaterthannu} holds true with $r^\nu=\max\{\rho,\rho_1,\ldots,\rho_{N}\}$.

Now, the definition of $\mathscr{C}_n(r^\nu)$ guarantees the existence of a sequence $(\bm{y}^\nu_n)_{n\in\mathbb{N}}$ of points in $\mathbb{R}^3$ satisfying
\begin{align*}
\bigl\|1_{\mathcal{B}(\bm{y}^\nu_n,r^\nu)}\psi_n\bigr\|_{L^2}^2>\lambda-\nu\quad\textrm{for all }n\in\mathbb{N}.
\end{align*}
Our aim will be to prove that by setting $\bm{y}_n=\bm{y}_n^{\frac{\lambda}{2}}$ for $n\in\mathbb{N}$ we obtain a sequence $(\bm{y}_n)_{n\in\mathbb{N}}$ with properties as stated in the proposition. To achieve this goal consider an arbitrary $\varepsilon$ in the interval $\bigl(0,\frac{\lambda}{2}\bigr)$. Then for every $n\in\mathbb{N}$ both of the integrals $\int_{\mathcal{B}(\bm{y}_n,r^{\lambda/2})}|\psi_n(\bm{x})|^2\,\mathrm{d}\bm{x}$ and $\int_{\mathcal{B}(\bm{y}_n^\varepsilon,r^{\varepsilon})}|\psi_n(\bm{x})|^2\,\mathrm{d}\bm{x}$ must be strictly larger than $\frac{\lambda}{2}$, which together with the fact that $\int_{\mathbb{R}^3}|\psi_n(\bm{x})|^2\,\mathrm{d}\bm{x}=\lambda$ gives that the balls $\mathcal{B}(\bm{y}_n,r^{\frac{\lambda}{2}})$ and $\mathcal{B}(\bm{y}_n^\varepsilon,r^{\varepsilon})$ have a nonempty intersection.
\begin{figure}
	\centering
		\includegraphics[width=0.6\textwidth]{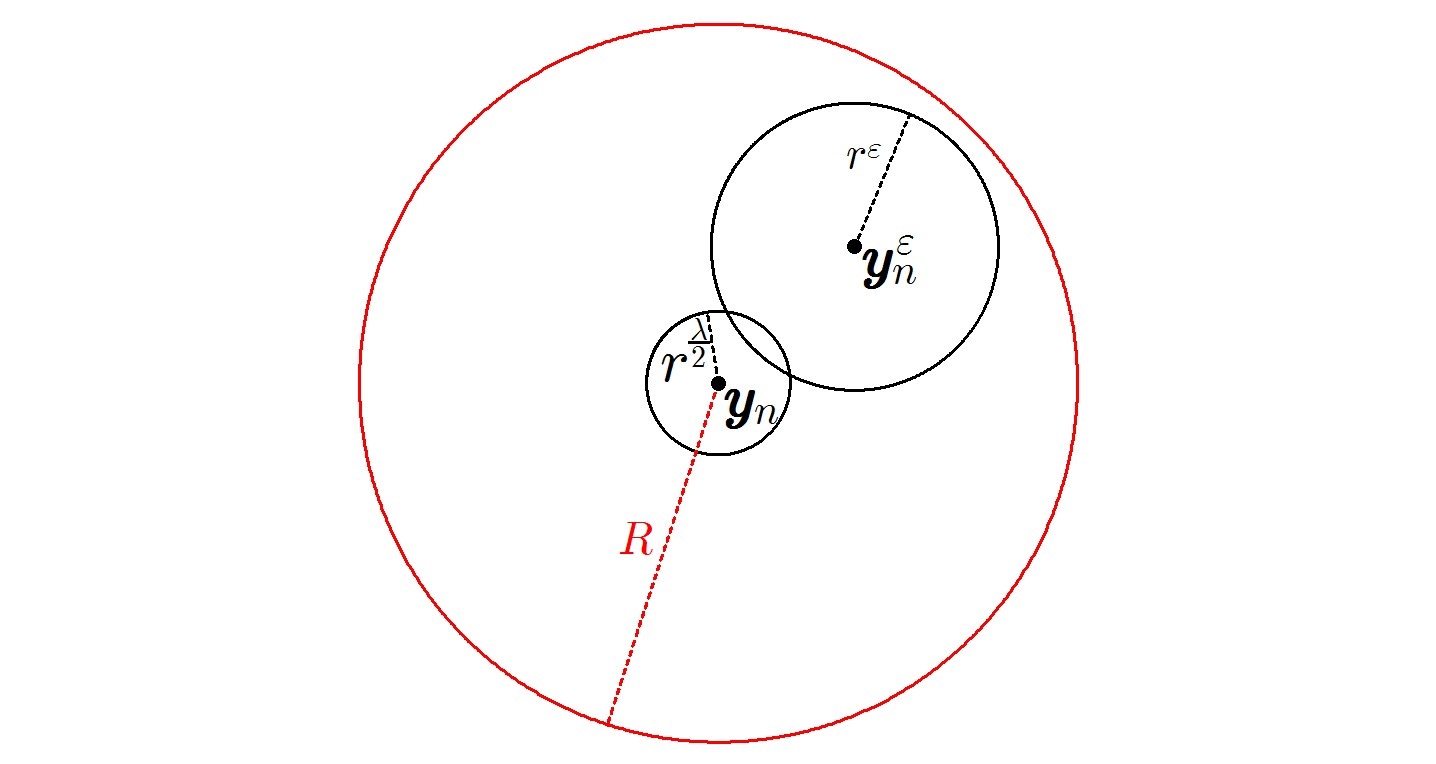}
	\caption{By setting $\textcolor{red}{R=r^{\frac{\lambda}{2}}+2r^{\varepsilon}}$ we get that $\mathcal{B}(\bm{y}_n^{\varepsilon},r^\varepsilon)\subset \mathcal{B}(\bm{y}_n,\textcolor{red}{R})$ because the balls $\mathcal{B}(\bm{y}_n,r^{\frac{\lambda}{2}})$ and $\mathcal{B}(\bm{y}_n^\varepsilon,r^{\varepsilon})$ are guaranteed not to be disjoint.}
	\label{fig:Ballstight}
\end{figure}
This enables us to define $R=r^{\frac{\lambda}{2}}+2r^{\varepsilon}$ and thereby obtain
\begin{align*}
\bigl\|1_{\mathcal{B}(\bm{y}_n,R)}\psi_n\bigr\|_{L^2}^2\geq \bigl\|1_{\mathcal{B}(\bm{y}_n^{\varepsilon},r^{\varepsilon})}\psi_n\bigr\|_{L^2}^2>\lambda-\varepsilon,
\end{align*}
for any $n\in\mathbb{N}$, which is the desired result.
\end{proof}

\section{The Lower Semicontinuity Argument}
We began by considering an arbitrary minimizing sequence $\bigl((\psi_n,\bm{A}_n)\bigr)_{n\in\mathbb{N}}$ for $\mathscr{E}_j^{\bm{v}}$. As we will see below our efforts in the previous section enable us to apply the direct method in the calculus of variations to the sequence of translated pairs
\begin{align*}
(\psi_n',\bm{A}_n')=(\psi_n\circ \tau_{\bm{y}_n},\bm{A}_n\circ \tau_{\bm{y}_n}),
\end{align*}
where $(\bm{y}_n)_{n\in\mathbb{N}}$ denotes the sequence whose existence is guaranteed in Proposition \ref{tightnesssaetning}. Due to $\mathscr{E}_j^{\bm{v}}$'s translation invariance $\bigl((\psi_n',\bm{A}_n')\bigr)_{n\in\mathbb{N}}$ will namely be a minimizing sequence for $\mathscr{E}_j^{\bm{v}}$ and by Proposition \ref{tightnesssaetning} we have
\begin{align}
\forall \varepsilon>0\,\exists R>0\forall n\in\mathbb{N}:\bigl\|1_{\mathcal{B}(\bm{0},R)}\psi_n'\bigr\|_{L^2}^2=\bigl\|1_{\mathcal{B}(\bm{y}_n,R)}\psi_n\bigr\|_{L^2}^2\geq \lambda-\varepsilon.\label{tightness}
\end{align}
This enables us to show the existence of a minimizer for $\mathscr{E}_j^{\bm{v}}$ on $\mathcal{S}_{\lambda}$.
\begin{saetning}\label{minimization}
For every choice of $j\in\{\mathrm{S},\mathrm{P}\}$, $\bm{v}\in\mathbb{R}^3$ with $0<|\bm{v}|<c$ and $\lambda\in \Lambda_j^{\bm{v}}$ there exists a pair $(\psi,\bm{A})\in \mathcal{S}_{\lambda}$ such that $\mathscr{E}_j^{\bm{v}}(\psi,\bm{A})=I_j^\lambda$.
\end{saetning}
\begin{proof}
According to Lemma \ref{minbdd} (together with the Sobolev inequality) the sequences $(\psi_n')_{n\in\mathbb{N}}$, $(\bm{A}_n')_{n\in\mathbb{N}}$ and $(\nabla\otimes\bm{A}_n')_{n\in\mathbb{N}}$ are bounded in the reflexive Banach spaces $H^1$, $L^6$ respectively $L^2$. Thus, the Banach-Alaoglu theorem gives the existence of functions $\psi\in H^1$ and $\bm{A}\in L^6$ with square integrable derivatives such that (passing to subsequences)
\begin{align}
\psi_n'\xrightharpoonup[n\to\infty]{}\psi\textrm{ in }H^1,\quad \bm{A}_n'\xrightharpoonup[n\to\infty]{}\bm{A}\textrm{ in }L^6\quad \textrm{and}\quad \partial_\ell\bm{A}_n'\xrightharpoonup[n\to\infty]{}\partial_\ell\bm{A}\textrm{ in }L^2\label{weakconvergence}
\end{align}
for $\ell\in\{1,2,3\}$. 
Observe that we can (after passing to yet another subsequence) assume that
\begin{align}
\psi_n'\xrightarrow[n\to\infty]{}\psi\textrm{ and }\bm{A}_n'\xrightarrow[n\to\infty]{}\bm{A}\textrm{ pointwise almost everywhere in }\mathbb{R}^3\label{pointwisepsiA}
\end{align}
as a consequence of \eqref{weakconvergence} and the result \cite[Corollary 8.7]{LL} about weak convergence implying a.e. convergence of a subsequence.

The pair $(\psi,\bm{A})$ is our candidate for a minimizer, so we begin by showing that $(\psi,\bm{A})\in \mathcal{S}_{\lambda}$. In this context the only nontrivial condition to check is that the identity $\|\psi\|_{L^2}^2=\lambda$ holds true. The inequality $\|\psi\|_{L^2}^2\leq \lambda$ follows immediately from \eqref{weakconvergence} and weak lower semicontinuity of the norm $\|\cdot\|_{L^2}$ (as expressed in \cite[Theorem 2.11]{LL}). To prove the opposite inequality we let $\varepsilon>0$ be given and use \eqref{tightness} to choose an $R>0$ such that
\begin{align}
\bigl\|1_{\mathcal{B}(\bm{0},R)}\psi_n'\bigr\|_{L^2}^2\geq \lambda-\varepsilon\quad\textrm{for all }n\in\mathbb{N}.\label{interball}
\end{align}
By \eqref{weakconvergence} and the Rellich-Kondrashov theorem \cite[Theorem 8.6]{LL} the left hand side of \eqref{interball} converges to $\bigl\|1_{\mathcal{B}(\bm{0},R)}\psi\bigr\|_{L^2}^2$ as $n$ tends to infinity. Hence we have $\|\psi\|_{L^2}^2\geq \lambda-\varepsilon$ for any $\varepsilon>0$ and consequently $\|\psi\|_{L^2}^2\geq \lambda$. Besides giving the desired conclusion that $(\psi,\bm{A})\in \mathcal{S}_{\lambda}$ this also enables us to deduce that
\begin{align}
\psi_n'\xrightarrow[n\to\infty]{}\psi\textrm{ in }L^2,\label{strongL2psi}
\end{align}
simply because the recently gained knowledge that $\|\psi\|_{L^2}^2=\|\psi_n'\|_{L^2}^2=\lambda$ gives together with \eqref{weakconvergence} that
\begin{align*}
\|\psi-\psi_n'\|_{L^2}^2=\|\psi\|_{L^2}^2+\|\psi_n'\|_{L^2}^2-2\mathrm{Re}\bigl(\psi,\psi_n'\bigr)_{L^2}\xrightarrow[n\to\infty]{}0.
\end{align*}
Finally, we observe that by \eqref{pointwisepsiA} the sequence $((A_n')^\ell\psi_n')_{n\in\mathbb{N}}$ converges pointwise almost everywhere to $A^\ell\psi$ and by the Hölder inequality it is bounded by $K_S^{\frac{3}{2}}C^{\frac{3}{2}}\lambda^{\frac{1}{4}}$ in $L^2$ for each $\ell\in\{1,2,3\}$. We now use that a bounded sequence of functions converging pointwise a.e. to some $L^2$-limit also converges weakly in $L^2$ to the same limit -- to prove the weak convergence it suffices namely to test against $C_0^\infty$-functions by \cite[Theorem V.1.3]{Yo} and thus the result follows from Egorov's theorem \cite[Section 0.3]{Yo}. This yields
\begin{align*}
(A'_n)^\ell\psi_n'\xrightharpoonup[n\to\infty]{}A^\ell\psi\textrm{ in }L^2,
\end{align*}
which together with \eqref{weakconvergence} implies that
\begin{align}
\nabla_{j,\bm{A}_n'}\psi_n'\xrightharpoonup[n\to\infty]{}\nabla_{j,\bm{A}}\psi\textrm{ in }L^2.\label{kineticliminf}
\end{align}

The remaining task to overcome is proving that $I_j^\lambda=\mathscr{E}_j^{\bm{v}}(\psi,\bm{A})$ -- or rather that $I_j^\lambda\geq \mathscr{E}_j^{\bm{v}}(\psi,\bm{A})$ since the opposite inequality is trivially true. By superadditivity of $\liminf$ and \eqref{Eexp1} we get
\begin{align}
I_j^\lambda&\geq\frac{1}{2m}\liminf_{n\to\infty}\|\nabla_{j,\bm{A}_n'}\psi_n'\|_{L^2}^2+\hbar\liminf_{n\to\infty}(\psi_n',i\bm{v}\cdot\nabla\psi_n')_{L^2}\nonumber\\
&+\frac{1}{8\pi}\liminf_{n\to\infty}\Bigl(\|\nabla\otimes\bm{A}_n'\|_{L^2}^2-\Bigl\|\Bigl(\frac{\bm{v}}{c}\cdot\nabla\Bigr)\bm{A}_n'\Bigr\|_{L^2}^2\Bigr),\label{liminfestimates}
\end{align}
so we will have to estimate each of the terms on the right hand side of \eqref{liminfestimates}. That
\begin{align}
\liminf_{n\to\infty}\|\nabla_{j,\bm{A}_n'}\psi_n'\|_{L^2}^2\geq \|\nabla_{j,\bm{A}}\psi\|_{L^2}^2\label{liminf1}
\end{align}
follows immediately from \eqref{kineticliminf} and the weak lower semicontinuity of $\|\cdot\|_{L^2}$. Using \eqref{strongL2psi}, \eqref{weakconvergence} and Lemma \ref{minbdd} gives
\begin{align*}
&\bigl|(\psi_n',i\bm{v}\cdot\nabla\psi_n')_{L^2}-(\psi,i\bm{v}\cdot\nabla\psi)_{L^2}\bigr|\\
&\leq |v|C\|\psi_n'-\psi\|_{L^2}+\sum_{\ell=1}^3|v^\ell||(\psi,\partial_\ell\psi_n'-\partial_\ell\psi)_{L^2}|\\
&\xrightarrow[n\to\infty]{} 0
\end{align*}
and therefore
\begin{align}
\liminf_{n\to\infty}(\psi_n',i\bm{v}\cdot\nabla\psi_n')_{L^2}=(\psi,i\bm{v}\cdot\nabla\psi)_{L^2}.\label{liminf2}
\end{align}
To treat the last term on the right hand side of \eqref{liminfestimates} we imagine that $\bm{v}$ points in the direction of the first axis, whereby the functional
\begin{align*}
\mathscr{H}:D^1\ni \bm{B}\mapsto \sqrt{\|\nabla\otimes\bm{B}\|_{L^2}^2-\Bigl\|\Bigl(\frac{\bm{v}}{c}\cdot\nabla\Bigr)\bm{B}\Bigr\|_{L^2}^2}\in \mathbb{R}
\end{align*}
simply reduces to
\begin{align*}
\mathscr{H}(\bm{B})=\left\|\begin{pmatrix}\sqrt{1-\frac{\bm{v}^2}{c^2}}\partial_1\\\partial_2\\\partial_3\end{pmatrix}\otimes\bm{B}\right\|_{L^2}.
\end{align*}
It is immediately apparent that $\mathscr{H}$ is a convex functional (by the triangle inequality) and $\mathscr{H}$ is continuous $D^1\to\mathbb{R}$ since $\mathscr{H}(\bm{B})\leq\|\nabla\otimes\bm{B}\|_{L^2}$ for all $\bm{B}\in D^1$. Therefore we get from Mazur's theorem \cite[Corollary 3.9]{HB} that $\mathscr{H}$ is weakly lower semicontinuous, which together with \eqref{weakconvergence} implies that
\begin{align}
\liminf_{n\to\infty}\Bigl(\|\nabla\otimes\bm{A}_n'\|_{L^2}^2-\Bigl\|\Bigl(\frac{\bm{v}}{c}\cdot\nabla\Bigr)\bm{A}_n'\Bigr\|_{L^2}^2\Bigr)\geq \|\nabla\otimes\bm{A}\|_{L^2}^2-\Bigl\|\Bigl(\frac{\bm{v}}{c}\cdot\nabla\Bigr)\bm{A}\Bigr\|_{L^2}^2,\label{liminf3}
\end{align}
since weak convergence of a sequence $(f_n)_{n\in\mathbb{N}}$ to some function $f$ in the Hilbert space $D^1$ by the Riesz representation theorem is characterized by the limit $\bigl(\nabla f_n,\nabla g\bigr)_{L^2}\xrightarrow[n\to\infty]{}\bigl(\nabla f,\nabla g\bigr)_{L^2}$ holding true for every choice $g\in D^1$. Finally, combining \eqref{liminfestimates}, \eqref{liminf1}, \eqref{liminf2} and \eqref{liminf3} results in the inequality $I_j^\lambda\geq \mathscr{E}_j^{\bm{v}}(\psi,\bm{A})$.
\end{proof}
Combining Theorem \ref{minimization} with Lemma \ref{Mintosolve} now gives the main result. 

\chapter{Behavior of the Energy for small velocities of the Particle}\label{behaviorofenergy}
In this final section we estimate the energy of our travelling wave solutions. The next Theorem \ref{energyact} shows that to leading order for small $|\bm{v}|$ the energy behaves like $\frac{m\bm{v}^2}{2}\lambda$. We interpret this as saying that there is no change in effective mass due to the electromagnetic field.
\begin{saetning}\label{energyact}
Let $j\in\{\mathrm{S},\mathrm{P}\}$ and $\lambda>0$ be given. Then there exist $\theta_j,\kappa_j>0$ \emph{(}only depending on $j,\lambda,\hbar,c,Q$ and $m$\emph{)} such that
\begin{align*}
\Bigl|E_j(\bm{v},\psi,\bm{A})-\frac{m\bm{v}^2}{2}\lambda\Bigr|\leq \kappa_j|\bm{v}|^3
\end{align*}
for any $\bm{v}\in\mathbb{R}^3$ with $0<|\bm{v}|<\theta_j$ and any minimizer $(\psi,\bm{A})$ of $\mathscr{E}_j^{\bm{v}}$ on $\mathcal{S}_{\lambda}$.
\end{saetning}
\begin{proof}
Let $j\in\{\mathrm{S},\mathrm{P}\}$, $\lambda>0$ as well as $\bm{v}\in\mathbb{R}^3$ with $0<|\bm{v}|<\Theta_{j,+}^{\lambda}$ be given and consider an arbitrary minimizer $(\psi,\bm{A})$ of $\mathscr{E}_j^{\bm{v}}$ on $\mathcal{S}_{\lambda}$. Then according to Lemma \ref{negativeI} and \eqref{boundonpsisecond} the pair $(\psi,\bm{A})$ must satisfy \eqref{FirstcasePaulilow}. Thereby \eqref{boundonA}, \eqref{boundonpsifirst} and Lemma \ref{negativeI} give that
\begin{align*}
\|\psi\|_{L^6}^2\leq \frac{2^6\pi^2 K_S^8Q^4m^2\lambda^3}{\hbar^4}\frac{\bm{v}^4}{(\Theta_{j,+}^{\lambda}-|\bm{v}|)^2(|\bm{v}|-\Theta_{j,-}^{\lambda})^2}
\end{align*}
and
\begin{align*}
\|\nabla\otimes\bm{A}\|_{L^2}^2\leq \frac{2^8\pi^3 K_S^6c^2Q^4m\lambda^3}{\hbar^2}\frac{\bm{v}^4}{(c^2-\bm{v}^2)^2(\Theta_{j,+}^{\lambda}-|\bm{v}|)(|\bm{v}|-\Theta_{j,-}^{\lambda})}.
\end{align*}
Using these estimates together with \eqref{Eexp2}, Lemma \ref{negativeI}, Hölder's and Sobolev's inequalities results in the inequality
\begin{align*}
\|\nabla_{j,\bm{A}+\frac{mc}{Q}\bm{v}}\psi\|_{L^2}^2\leq \frac{2^6\pi^2K_S^6Q^4m^2\lambda^3}{\hbar^2}\frac{\bm{v}^4\bigl(c^2\bigl(1+\sqrt{2}\bigr)+\bm{v}^2\bigl(1-\sqrt{2}\bigr)\bigr)}{(c^2-\bm{v}^2)^2(\Theta_{j,+}^{\lambda}-|\bm{v}|)(|\bm{v}|-\Theta_{j,-}^{\lambda})}
\end{align*}
and so the desired result follows immediately from the identities
\begin{align*}
E_{\mathrm{S}}(\bm{v},\psi,\bm{A})=\frac{1}{2m}\|\nabla_{\mathrm{S},\bm{A}+\frac{mc}{Q}\bm{v}}\psi\|_{L^2}^2+\frac{m\bm{v}^2}{2}\lambda-(\psi,\bm{v}\cdot\nabla_{\mathrm{S},\bm{A}+\frac{mc}{Q}\bm{v}}\psi)_{L^2}\\
+\frac{1}{8\pi}\int_{\mathbb{R}^3}\Bigl(\Bigl|\Bigl(\frac{v}{c}\cdot\nabla\Bigr)\bm{A}\Bigr|^2+|\nabla\times \bm{A}|^2\Bigr)\,\mathrm{d}\bm{x}\lambda
\end{align*}
and
\begin{align*}
E_{\mathrm{P}}(\bm{v},\psi,\bm{A})=\frac{1}{2m}\|\nabla_{\mathrm{P},\bm{A}+\frac{mc}{Q}\bm{v}}\psi\|_{L^2}^2+\frac{m\bm{v}^2}{2}\lambda-\mathrm{Re}(\bm{\sigma}\cdot\bm{v}\psi,\nabla_{\mathrm{P},\bm{A}+\frac{mc}{Q}\bm{v}}\psi)_{L^2}\\
+\frac{1}{8\pi}\int_{\mathbb{R}^3}\Bigl(\Bigl|\Bigl(\frac{v}{c}\cdot\nabla\Bigr)\bm{A}\Bigr|^2+|\nabla\times \bm{A}|^2\Bigr)\,\mathrm{d}\bm{x}\lambda.
\end{align*}
\end{proof}

\appendix

\chapter{The Poisson Equation}\label{PoissonAppendix}
Given some function $f$ the corresponding Poisson equation reads
\begin{align}
-\Delta u=f.\label{Poisson}
\end{align}
Let us briefly recall the contents of \cite[Theorem 6.21]{LL} and \cite[Remark 6.21(2)]{LL}: If
\begin{align}
f\in L^1_{\mathrm{loc}}(\mathbb{R}^3)\quad\textrm{and}\quad\int_{\mathbb{R}^3} \frac{|f(\bm{y})|}{1+|\bm{y}|}\,\mathrm{d}\bm{y}<\infty\label{integrabilityf}
\end{align}
then defining $u:\mathbb{R}^3\to\mathbb{C}$ by
\begin{align}
u(\bm{x})=\frac{1}{4\pi}\int_{\mathbb{R}^3}\frac{1}{|\bm{x}-\bm{y}|}f(\bm{y})\,\mathrm{d}\bm{y}\label{Poissonsol}
\end{align}
for almost every $\bm{x}\in\mathbb{R}^3$ results in a locally integrable solution of \eqref{Poisson}. Moreover, the distributional gradient $\nabla u$ can be identified with the function given by
\begin{align}
\nabla u(\bm{x})=-\frac{1}{4\pi}\int_{\mathbb{R}^3}\frac{\bm{x}-\bm{y}}{|\bm{x}-\bm{y}|^3}f(\bm{y})\,\mathrm{d}\bm{y}
\label{derivPoisson}
\end{align}
for almost every $\bm{x}\in\mathbb{R}^3$. We will need the following result.
\begin{lemma}\label{Poissonlemma}
If $f\in L^1\cap L^3$ and $\nabla f\in L^1\cap L^{\frac{5}{4}}$ then $u$ defined by \eqref{Poissonsol} is a $D^1$-function solving \eqref{Poisson} in the distribution sense. Likewise, if $f\in H^1$ has compact support then $u$ solves \eqref{Poisson} and $\nabla u\in D^1$.
\end{lemma}
\begin{proof}
Verifying the condition \eqref{integrabilityf} in each of the two scenarios outlined in the statement of the lemma is an easy task, which is left for the reader -- in this context it is useful to note that $\bm{y}\mapsto \frac{1}{1+|\bm{y}|}$ is e.g. an $L^6$-function. Thus, the function defined almost everywhere by \eqref{Poissonsol} is indeed a solution of the Poisson equation in those two cases.

Suppose that $f\in L^1\cap L^3$ and $\nabla f\in L^1\cap L^{\frac{5}{4}}$. Then we first show that $u$ vanishes at infinity: For this let $\mathcal{N}$ denote the null set on which the identities \eqref{Poissonsol} and \eqref{derivPoisson} do not hold true. Then for each sequence $(\bm{x}_k)_{k\in\mathbb{N}}$ of elements in $\mathbb{R}^3\setminus \mathcal{N}$ with $|\bm{x}_k|\xrightarrow[k\to\infty]{} \infty$ we have
\begin{align*}
|u(\bm{x}_k)|&\leq\frac{1}{4\pi}\bigl\|1_{\mathcal{B}(\bm{x}_k,1)}f\bigr\|_{L^3}\left\|\frac{1_{\mathcal{B}(\bm{0},1)}}{|\cdot|}\right\|_{L^{\frac{3}{2}}}+\frac{1}{4\pi}\left\|\frac{(1_{\mathbb{R}^3\setminus\mathcal{B}(\bm{x}_k,1)}f)(\cdot)}{|\bm{x}_k-\cdot\,|}\right\|_{L^1}\\
&\xrightarrow[k\to\infty]{} 0,
\end{align*}
where we split the integral involved in the expression for $u(\bm{x}_k)$ into a contribution from $\mathcal{B}(\bm{x}_k,1)$ as well as a contribution from $\mathbb{R}^3\setminus\mathcal{B}(\bm{x}_k,1)$ and treat these by means of the Hölder inequality and Lebesgue's dominated convergence theorem. 
In order to prove that $\nabla u$ is square integrable we use the Hölder inequality and Tonelli's theorem to get
\begin{align*}
&\int_{\mathbb{R}^3}|\nabla u(\bm{x})|^2\,\mathrm{d}\bm{x}\\
&\leq \frac{1}{8\pi^2}\int_{\mathbb{R}^3}\Bigl(\int_{\mathcal{B}(\bm{x},1)}\frac{|f(\bm{y})|}{|\bm{x}-\bm{y}|^2}\,\mathrm{d}\bm{y}\Bigr)^2\,\mathrm{d}\bm{x}+\frac{1}{8\pi^2}\int_{\mathbb{R}^3}\Bigl(\int_{\mathbb{R}^3\setminus\mathcal{B}(\bm{x},1)}\frac{|f(\bm{y})|}{|\bm{x}-\bm{y}|^2}\,\mathrm{d}\bm{y}\Bigr)^2\,\mathrm{d}\bm{x}\\
&\leq\frac{1}{8\pi^2}\int_{\mathbb{R}^3}\int_{\mathcal{B}(\bm{x},1)}\frac{|f(\bm{y})|}{|\bm{x}-\bm{y}|^{\frac{5}{2}}}\,\mathrm{d}\bm{y}\,\mathrm{d}\bm{x}\left\|\frac{1_{\mathcal{B}(\bm{0},1)}}{|\cdot|}\right\|_{L^{\frac{5}{2}}}^{\frac{3}{2}}\|f\|_{L^1}^{\frac{1}{10}}\|f\|_{L^3}^{\frac{9}{10}}\\
&+\frac{1}{8\pi^2}\int_{\mathbb{R}^3}\int_{\mathbb{R}^3\setminus\mathcal{B}(\bm{x},1)}\frac{|f(\bm{y})|}{|\bm{x}-\bm{y}|^{\frac{7}{2}}}\,\mathrm{d}\bm{y}\,\mathrm{d}\bm{x}\left\|\frac{1_{\mathbb{R}^3\setminus\mathcal{B}(\bm{0},1)}}{|\cdot|}\right\|_{L^{\frac{7}{2}}}^{\frac{1}{2}}\|f\|_{L^3}^{\frac{3}{14}}\|f\|_{L^1}^{\frac{11}{14}}\\
&\leq \frac{1}{8\pi^2}\left\|\frac{1_{\mathcal{B}(\bm{0},1)}}{|\cdot|}\right\|_{L^{\frac{5}{2}}}^{4}\|f\|_{L^1}^{\frac{11}{10}}\|f\|_{L^3}^{\frac{9}{10}}+\frac{1}{8\pi^2}\left\|\frac{1_{\mathbb{R}^3\setminus\mathcal{B}(\bm{0},1)}}{|\cdot|}\right\|_{L^{\frac{7}{2}}}^{4}\|f\|_{L^3}^{\frac{3}{14}}\|f\|_{L^1}^{\frac{25}{14}}
\end{align*}
and so we conclude that $u\in D^1$.

Assume now that $f\in H^1$ has support in some ball $\mathcal{B}(\bm{0},r_f)$. Then for all $\bm{x}\in\mathbb{R}^3\setminus \mathcal{N}$ with $|\bm{x}|>r_f$ we have
\begin{align*}
|\nabla u(\bm{x})|\leq \frac{1}{4\pi(|\bm{x}|-r_f)^2}\Bigl(\frac{4}{3}\pi r_f^3\Bigr)^{\frac{1}{2}}\|f\|_{L^2}
\end{align*}
whereby we deduce that $\nabla u$ vanishes at infinity. Combining the change of variables $\bm{z}=\bm{x}-\bm{y}$ with a naive differentiation under the integral sign in \eqref{derivPoisson} suggests that
\begin{align}
\partial_j\partial_k u(\bm{x})=-\frac{1}{4\pi}\int_{\mathbb{R}^3}\frac{x_k-y_k}{|\bm{x}-\bm{y}|^3} \partial_j f(\bm{y})\,\mathrm{d}\bm{y}\label{secondorderderiv}
\end{align}
for $j,k\in\{1,2,3\}$ and almost every $\bm{x}\in\mathbb{R}^3$. Under the assumption that $f$ has square integrable first derivatives and compact support, the right hand side of \eqref{secondorderderiv} is indeed well defined almost everywhere in $\mathbb{R}^3$: The function $\bm{y}\mapsto \frac{|\partial_j f(\bm{y})|}{|\bm{x}-\bm{y}|^2}$ (and thereby also $\bm{y}\mapsto \frac{x_k-y_k}{|\bm{x}-\bm{y}|^3} \partial_j f(\bm{y})$) must namely be integrable for almost all $\bm{x}\in\mathcal{B}(\bm{0},2r_f)$, since Tonelli's theorem, the Cauchy-Schwarz inequality and the basic observation stated on Figure \ref{fig:Integralestimatea} give that
\begin{align}
\int_{\mathcal{B}(\bm{0},2r_f)}\int_{\mathcal{B}(\bm{0},r_f)}\frac{|\partial_j f(\bm{y})|}{|\bm{x}-\bm{y}|^2}\,\mathrm{d}\bm{y}\,\mathrm{d}\bm{x}
&\leq \Bigl(\frac{4}{3}\pi r_f^3\Bigr)^{\frac{1}{2}}\left\|1_{\mathcal{B}(\bm{0},3r_f)}\frac{1}{|\cdot|}\right\|_{L^2}^2\|\partial_j f\|_{L^2}.\label{locLaplace}
\end{align}
On the other hand $\bm{y}\mapsto \frac{|\partial_j f(\bm{y})|}{|\bm{x}-\bm{y}|^2}$ is majorized by the integrable function $\frac{|\partial_j f|}{r_f^2}$ for all $\bm{x}\in\mathbb{R}^3\setminus\mathcal{B}(\bm{0},2r_f)$, as shown on Figure \ref{fig:Integralestimateb}.
\begin{figure}
\centering
\subbottom[$\textcolor{red}{|\bm{x}-\bm{y}|< 3r_f}$]{\includegraphics[width=0.45\textwidth]{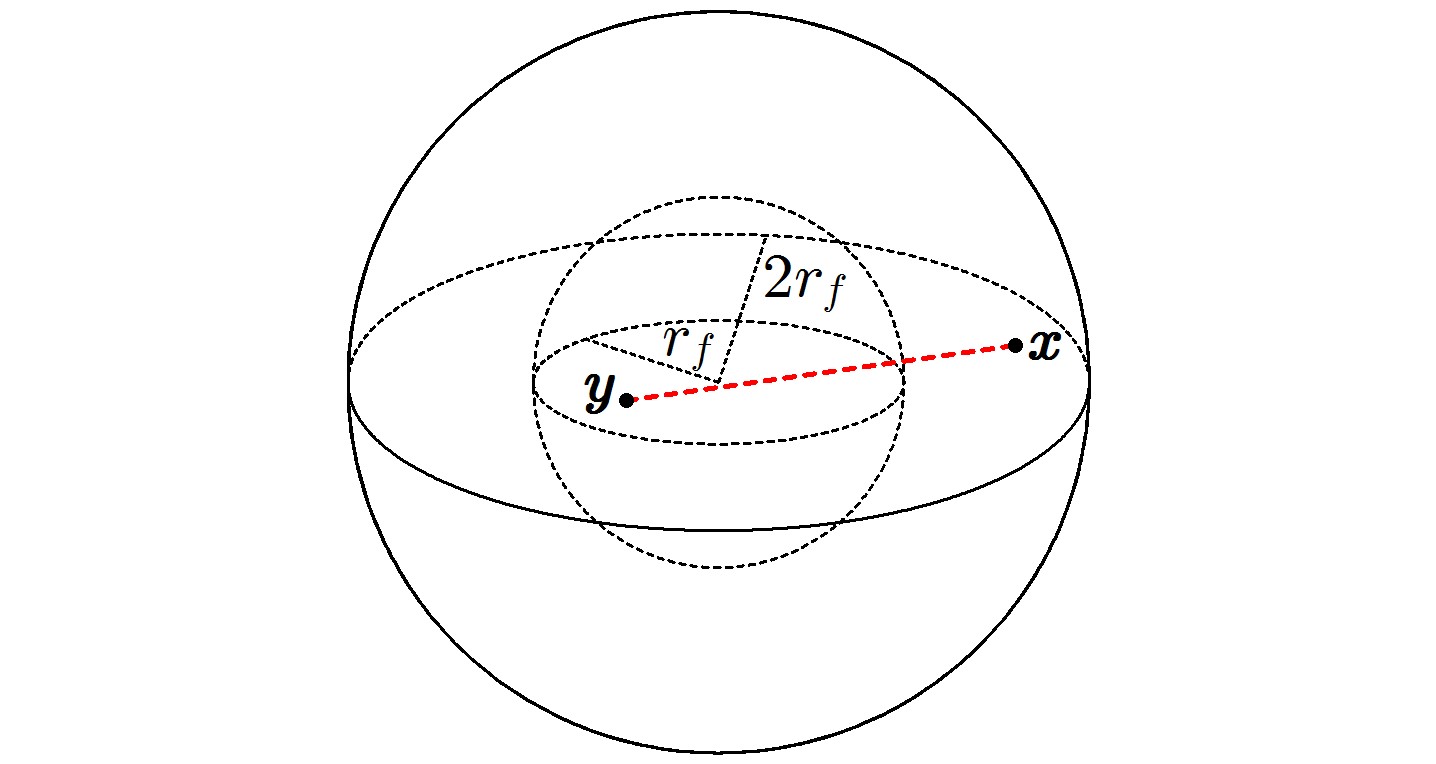}\label{fig:Integralestimatea}}
\hfill
\subbottom[$\textcolor{red}{|\bm{x}-\bm{y}|> r_f}$]{\includegraphics[width=0.45\textwidth]{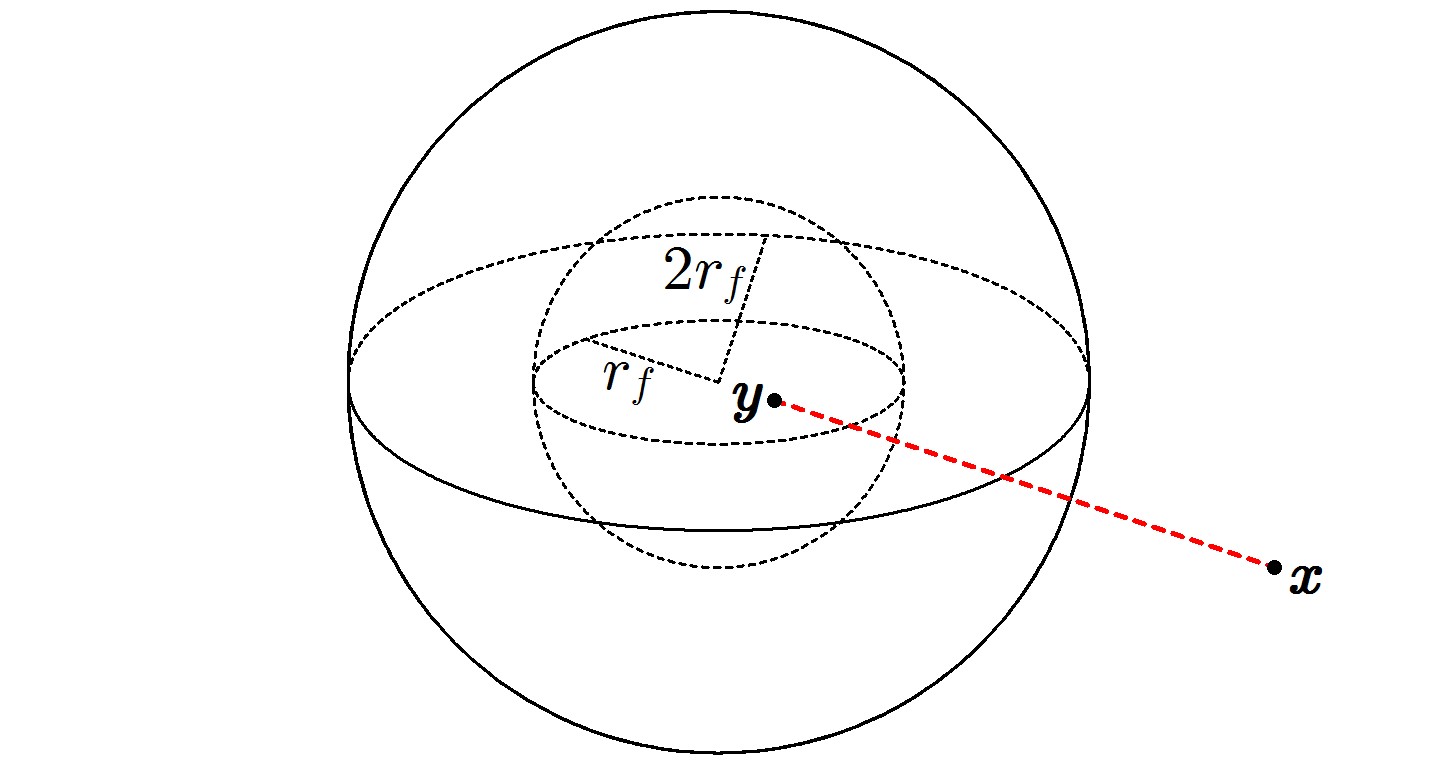}\label{fig:Integralestimateb}}
\caption{Estimates on the distance from $\bm{x}$ to any $\bm{y}\in \mathrm{supp}f\subset \mathcal{B}(\bm{0},r_f)$ in the cases $\bm{x}\in\mathcal{B}(\bm{0},2r_f)$ respectively $\bm{x}\in\mathbb{R}^3\setminus\mathcal{B}(\bm{0},2r_f)$. Both estimates follow from the triangle inequality in $\mathbb{R}^3$.}
\label{fig:Integralestimate}
\end{figure}
Thus, it makes sense to consider the function given (a.e.) by the expression on the right hand side of \eqref{secondorderderiv} -- the identification of this function with the distribution $\partial_j\partial_k u$ then follows from a standard argument utilizing Fubini's theorem (which is outlined in the proof of \cite[Theorem 6.21]{LL}). It just remains to be proven that the functions $\partial_j\partial_k u$ are square integrable -- we first verify this square integrability on the ball $\mathcal{B}(\bm{0},2r_f)$. This is done by using the Hölder inequality, the Jensen inequality and the Hardy-Littlewood-Sobolev inequality:
\begin{align*}
&\int_{\mathcal{B}(\bm{0},2r_f)}|\partial_j\partial_k u(\bm{x})|^2\,\mathrm{d}\bm{x}\\
&\leq \!\frac{1}{16\pi^2}\!\!\int_{\mathcal{B}(\bm{0},2r_f)}\!\Bigl(\int_{\mathcal{B}(\bm{0},r_f)}\!\!\frac{|\partial_j f(\bm{y})|^{\frac{6}{5}}}{|\bm{x}-\bm{y}|^{\frac{9}{4}}}\mathrm{d}\bm{y}\Bigr)^{\frac{10}{9}}\!\Bigl(\int\frac{1_{\mathcal{B}(\bm{0},3r_f)}(\bm{x}-\bm{y})}{|\bm{x}-\bm{y}|^{\frac{27}{10}}}\mathrm{d}\bm{y}\Bigr)^{\frac{5}{9}}\!\|\partial_jf\|_{L^2}^{\frac{2}{3}}\mathrm{d}\bm{x}\\
&\leq \frac{\bigl(\frac{4}{3}\pi r_f^3\bigr)^{\frac{1}{9}}}{16\pi^2}\int_{\mathcal{B}(\bm{0},2r_f)}\int_{\mathcal{B}(\bm{0},r_f)}\frac{|\partial_j f(\bm{y})|^{\frac{4}{3}}}{|\bm{x}-\bm{y}|^{\frac{5}{2}}}\,\mathrm{d}\bm{y}\,\mathrm{d}\bm{x}\left\|1_{\mathcal{B}(\bm{0},3r_f)}\frac{1}{|\cdot|}\right\|_{L^{\frac{27}{10}}}^{\frac{3}{2}}\|\partial_jf\|_{L^2}^{\frac{2}{3}}\\
&\leq \frac{\bigl(\frac{4}{3}\pi r_f^3\bigr)^{\frac{1}{9}}}{8\pi^2} \Bigl(\frac{4}{3}\pi\Bigr)^{\frac{5}{6}}\Bigl(\Bigl(\frac{5}{3}\Bigr)^{\frac{5}{6}}+\Bigl(\frac{5}{2}\Bigr)^{\frac{5}{6}}\Bigr)\Bigl(\frac{4}{3}\pi (2r_f)^3\Bigr)^{\frac{1}{2}}\left\|1_{\mathcal{B}(\bm{0},3r_f)}\frac{1}{|\cdot|}\right\|_{L^{\frac{27}{10}}}^{\frac{3}{2}}\|\partial_jf\|_{L^2}^{2}.
\end{align*}
Finally, the Cauchy-Schwarz inequality, Tonelli's theorem and the observation on Figure \ref{fig:Integralestimateb} give
\begin{align*}
\int_{\mathbb{R}^3\setminus\mathcal{B}(\bm{0},2r_f)}|\partial_j\partial_k u(\bm{x})|^2\,\mathrm{d}\bm{x}
&\leq \frac{1}{16\pi^2}\Bigl(\frac{4}{3}\pi r_f^3\Bigr)\left\|1_{\mathbb{R}^3\setminus\mathcal{B}(\bm{0},r_f)}\frac{1}{|\cdot|}\right\|_{L^4}^4\|\partial_j f\|_{L^2}^2,
\end{align*}
whereby $\partial_j\partial_k u$ is also square integrable on $\mathbb{R}^3\setminus\mathcal{B}(\bm{0},2r_f)$. Consequently, $\nabla u$ is a $D^1$-function.
\end{proof}
\begin{bemaerkning}\label{uniquePoisson}
Consider a locally integrable, harmonic function $u$ with square integrable first derivatives. The harmonicity of $\nabla u$ ensures the existence of vector fields $\bm{p}_m$ on $\mathbb{R}^3$ with homogeneous harmonic polynomials of degree $m$ as coordinates such that
\begin{align*}
\nabla u(\bm{x})=\sum_{m=0}^\infty \bm{p}_m(\bm{x})
\end{align*}
for all $\bm{x}\in\mathbb{R}^3$ (see \cite[Corollary 5.34 and Proposition 1.30]{ABR}). The series even converges absolutely and uniformly on compact subsets of $\mathbb{R}^3$ so for an arbitrary given $R>0$ we have the series representation $\nabla u=\sum_{m=0}^\infty \bm{p}_m$ in $\bigl[L^2\bigl(\overline{\mathcal{B}}(\bm{0},R)\bigr)\bigr]^3$. Integrating in polar coordinates and using the homogeneity of the functions $\bm{p}_m$ as well as the spherical harmonic decomposition \cite[Theorem 5.12]{ABR} of $L^2(\partial \mathcal{B}(\bm{0},1))$ now gives
\begin{align}
\bigl\|1_{\overline{\mathcal{B}}(\bm{0},R)}\nabla u\bigr\|_{L^2}^2&=\sum_{m=0}^\infty \sum_{\ell=0}^\infty \int_0^R r^{m+\ell+2}\,\mathrm{d}r\bigl(\bm{p}_m,\bm{p}_\ell\bigr)_{L^2(\partial \mathcal{B}(\bm{0},1))}\nonumber\\
&=\sum_{m=0}^\infty \frac{R^{2m+3}}{2m+3}\|\bm{p}_m\|_{L^2(\partial\mathcal{B}(\bm{0},1))}^2.\label{harmonicR}
\end{align}
By Lebesgue's dominated convergence theorem the left hand side of \eqref{harmonicR} converges to $\|\nabla u\|_{L^2}^2$ as $R\to \infty$ so the same must be true for the right hand side. But the right hand side simply can not converge as $R\to\infty$ unless $\bm{p}_m\equiv\bm{0}$ for all $m\in\mathbb{N}_0$ -- so we conclude that $\nabla u\equiv\bm{0}$. Therefore $u$ is a constant function -- consequently, the Poisson equation can at most have one solution in the space $D^1$.
\end{bemaerkning}

\renewcommand{\bibname}{References}
\bibliographystyle{plain}
\bibliography{Bibliography}
\end{document}